\documentclass[11pt]{article}
\usepackage{amssymb}
\usepackage{amsthm}
\usepackage{amsmath}
\usepackage[dvips]{epsfig}
\usepackage{graphicx}
\usepackage{color}
\usepackage{url}

\usepackage[utf8]{inputenc}
\usepackage[left=1cm, right=2cm, top=2cm, bottom=2cm]{geometry}
\usepackage[english]{babel}
\usepackage{amsfonts}
\usepackage[numbers]{natbib}
\newtheorem{theorem}{\bf Theorem}
\newtheorem{lemma}{\bf Lemma}
\newtheorem{remark}{\bf Remark}
\newtheorem{definition}{\bf Definition}
\newtheorem{proposition}{\bf Proposition}

\newcommand{\PT}{{\cal PT}}

\newcommand{\ltz}{{l^2}}

\newcommand{\scell}[2][c]{\begin{tabular}[#1]{@{}c@{}}#2\end{tabular}}

\title{\bf Breathers in Hamiltonian $\PT$-symmetric chains \\ of coupled pendula under a resonant periodic force}

\author{Alexander Chernyavsky$^{1}$ and Dmitry E. Pelinovsky$^{1,2}$ \\
{\small $^{1}$ Department of Mathematics, McMaster University, Hamilton, Ontario, L8S 4K1, Canada} \\
{\small $^{2}$ Department of Applied Mathematics, Nizhny Novgorod State Technical University, Nizhny Novgorod, Russia }}

\begin{document}

\maketitle

\begin{abstract}
We derive a Hamiltonian version of the $\PT$-symmetric discrete nonlinear Schr\"{o}dinger equation
that describes synchronized dynamics of coupled pendula driven by a periodic movement of their common strings.
In the limit of weak coupling between the pendula, we classify the existence and spectral stability of breathers (time-periodic
solutions localized in the lattice) supported near one pair of coupled pendula. Orbital stability or instability of
breathers is proved in a subset of the existence region.
\end{abstract}

\section{Introduction}

Synchronization is a dynamical process where two or more interacting oscillatory systems end up with identical
movement. In 1665 Huygens experimented with maritime pendulum clocks and discovered the
anti-phase synchronization of two pendulums clocks mounted on the common frame~\cite{huygens}.
Since then, synchronization has become a basic concept in nonlinear and complex systems \cite{blekhman}.
Such systems include, but are not limited to, musical instruments, electric power systems, and
lasers. There are numerous applications in mechanical~\cite{nijmeijer} and electrical~\cite{yasser} engineering.
New applications are found in mathematical biology such as synchronous variation of cell nuclei, firing of
neurons, forms of cooperative behavior of animals and humans~\cite{pikovsky}.

Recently, Huygen's experiment has been widely discussed and several experimental devices were
built~\cite{bennett,kumon,pantaleone}. It was shown that two real mechanical clocks when mounted to a horizontally
moving beam can synchronize both in-phase and anti-phase~\cite{czolczynski}. In all these experiments
synchronization was achieved due to energy transfer via the oscillating beam, supporting Huygen's
intuition~\cite{bennett}.

One of the rapidly developing areas in between physics and mathematics, is the topic of $\mathcal{PT}$-symmetry,
which has started as a way to characterize non-Hermitian Hamiltonians in quantum mechanics~\cite{bender}.
The key idea is that a linear Schr{\"o}dinger operator with a complex-valued potential,
which is symmetric with respect to combined parity ($\mathcal{P}$) and
time-reversal($\mathcal{T}$) transformations, may have a real spectrum up to a certain critical value
of the complex potential amplitude. In nonlinear systems, this distinctive feature may
lead to existence of breathers (time-periodic solutions localized in space) as continuous families
of their energy parameter.

The most basic configuration having $\mathcal{PT}$ symmetry is a dimer, which represents a system of
two coupled oscillators, one of which has damping losses and the other one gains some energy from external sources.
This configuration was studied in numerous laboratory experiments involving electric circuits~\cite{schindlerExp},
superconductivity~\cite{rubinsteinExp}, optics~\cite{barashenkovExp,ruterExp} and microwave cavities~\cite{bittnerExp}.

In the context of synchronization of coupled oscillators in a $\mathcal{PT}$-symmetric system,
one of the recent experiments was performed by Bender et al.~\cite{benderExp}. These authors
considered a $\mathcal{PT}$-symmetric Hamiltonian system
describing the motion of two coupled pendula whose bases were connected by a horizontal rope
which moves periodically in resonance with the pendula.
The phase transition phenomenon, which is typical for $\mathcal{PT}$-symmetric
systems, happens when some of the real eigenvalues of the complex-valued Hamiltonian
become complex. The latter regime is said to have \emph{broken} $\mathcal{PT}$ symmetry.

On the analytical side, dimer equations were found to be completely integrable~\cite{barashenkovDimer,susanto}.
Integrability of dimers is obtained by using Stokes variables and it is lost when more coupled
nonlinear oscillators are added into a $\PT$-symmetric system. Nevertheless, it was understood recently
\cite{barashenkov,barashenkovInt} that there is a remarkable class of $\PT$-symmetric dimers
with cross-gradient Hamiltonian structure, where the real-valued Hamiltonians exist
both in  finite and infinite chains of coupled nonlinear oscillators.
Analysis of synchronization in the infinite chains of coupled oscillators in such class of models is
a subject of this work.

In the rest of this section, we describe how this paper is organized. We
also describe the main findings obtained in this work.

Section \ref{model} introduces the main model of coupled pendula driven by a resonant 
periodic movement of their common strings. See Figure \ref{pendulapic} for a schematic picture.
By using an asymptotic multi-scale method, the oscillatory dynamics of coupled pendula is reduced
to a $\PT$-symmetric discrete nonlinear Schr\"{o}dinger (dNLS) equation with gains and losses.
This equation generalizes the dimer equation derived in \cite{barashenkovInt,benderExp}.

Section \ref{symmetry} describes symmetries and conserved quantities for the $\PT$-symmetric dNLS equation.
In particular, we show that the cross-gradient Hamiltonian structure obtained in \cite{barashenkov,barashenkovInt}
naturally appears in the asymptotic reduction of the original Hamiltonian structure
of Newton's equations of motion for the coupled pendula.

Section \ref{breathers} is devoted to characterization of breathers, which are time-periodic
solutions localized in the chain. We show that depending on parameters
of the model (such as detuning frequency, coupling constant, driving force amplitude),
there are three possible types of breather solutions. For the first type, breathers of small
and large amplitudes are connected to each other and do not extend to symmetric
synchronized oscillations of coupled pendula. In the second and third types,
large-amplitude and small-amplitude breathers are connected to the symmetric synchronized oscillations
but are not connected to each other. See Figure \ref{branches} with branches (a), (b), and (c),
where the symmetric synchronized oscillations correspond to the value $E = 0$ and the
breather amplitude is given by parameter $A$.

Section \ref{zero-equilibrium} contains a routine analysis of linear stability of the zero equilibrium,
where the phase transition threshold to the broken $\PT$-symmetry phase is explicitly found.
Breathers are only studied for the parameters where the zero equilibrium is linearly stable.

Section \ref{variational} explores the Hamiltonian structure of the $\PT$-symmetric dNLS equation
and characterizes breathers obtained in Section \ref{breathers} from their energetic point of view.
We show that the breathers for large value of parameter $E$ appear to be saddle points of the Hamiltonian function
between continuous spectra of positive and negative energy, similar to the standing waves in the Dirac models.
Therefore, it is not clear from the energetic point of view if such breathers are linearly or nonlinearly stable.
On the other hand, we show that the breathers for smaller values of parameter $E$ appear to be saddle points
of the Hamiltonian function with a negative continuous spectrum and finitely many (either three or one) positive eigenvalues.

Section \ref{stability} is devoted to analysis of spectral and orbital stability of breathers.
For spectral stability, we use the limit of small coupling constant between the oscillators
(the same limit is also used in Sections \ref{breathers} and \ref{variational})
and characterize eigenvalues of the linearized operator. The main analytical results are
also confirmed numerically. See Figure \ref{spectra} for the three types of breathers.
Depending on the location of the continuous spectral bands
relative to the location of the isolated eigenvalues, we are able to prove nonlinear
orbital stability of breathers for branches (b) and (c).
We are also able to characterize instabilities of these types of breathers that
emerge depending on parameters of the model. Regarding branch (a), nonlinear stability analysis
is not available by using the energy method. Our follow-up work
\cite{ChernPel} develops a new method of analysis to prove the
long-time stability of breathers for branch (a).

The summary of our findings is given in the concluding Section \ref{conclusion},
where the main results are shown in the form of Table 1.

\section{Model}
\label{model}

A simple yet universal model widely used to study coupled nonlinear oscillators is
the Frenkel-Kontorova (FK) model \cite{kontorova}. It describes a chain of classical particles coupled to their neighbors
and subjected to a periodic on-site potential. In the continuum approximation, the FK
model reduces to the sine-Gordon equation, which is exactly integrable. The FK model is known to describe
a rich variety of important nonlinear phenomena, which find applications in solid-state physics and
nonlinear science~\cite{kivshar}.

We consider here a two-array system of coupled pendula, where each pendulum is
connected to the nearest neighbors by linear couplings. Figure \ref{pendulapic} shows
schematically that each array of pendula is connected in the longitudinal direction by the torsional
springs, whereas each pair of pendula is connected in the transverse direction by a common string.
Newton's equations of motion are given by
\begin{eqnarray}
\left\{ \begin{array}{l} \ddot{x}_n + \sin(x_n) = C \left( x_{n+1} - 2x_n + x_{n-1} \right) + D y_n, \\
 \ddot{y}_n + \sin(y_n) = C \left( y_{n+1} - 2 y_n + y_{n-1} \right) + D x_n,
\end{array} \right. \quad n \in \mathbb{Z}, \quad t \in \mathbb{R},
\label{oscillators}
\end{eqnarray}
where $(x_n,y_n)$ correspond to the angles of two arrays of pendula, dots denote derivatives of angles with respect
to time $t$, and the positive parameters $C$ and $D$ describe couplings between the two arrays
in the longitudinal and transverse directions, respectively. 
The type of coupling between the two pendula with the angles $x_n$ and $y_n$ 
is referred to as the {\em direct coupling} between nonlinear oscillators (see Section 8.2 in \cite{pikovsky}).

We consider oscillatory dynamics of coupled pendula under the following assumptions.
\begin{itemize}
\item[(A1)] The coupling parameters $C$ and $D$ are small. Therefore, we can introduce a small parameter $\mu$ such that
both $C$ and $D$ are proportional to $\mu^2$.

\item[(A2)] A resonant periodic force is applied to the common strings for each pair of coupled pendula.
Therefore, $D$ is considered to be proportional to $\cos(2 \omega t)$, where $\omega$ is selected
near the unit frequency of linear pendula indicating the $1 : 2$ parametric resonance between the force and the pendula.
\end{itemize}

\begin{figure}[h!]
\centering
\includegraphics[scale=0.75]{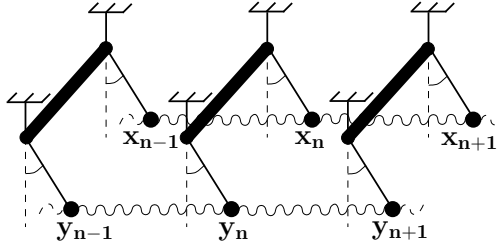}
\caption{A schematic picture for the chain of coupled pendula connected by torsional springs, where each pair is hung on a common string.}
\label{pendulapic}
\end{figure}

Mathematically, we impose the following representation for parameters $C$ and $D(t)$:
\begin{equation}
C = \epsilon \mu^2, \quad D(t) = 2 \gamma \mu^2 \cos(2 \omega t), \quad \omega^2 = 1 + \mu^2 \Omega,
\end{equation}
where $\gamma, \epsilon, \Omega$ are $\mu$-independent parameters, whereas
$\mu$ is the formal small parameter to characterize the two assumptions (A1) and (A2).

In the formal limit $\mu \to 0$, the pendula are uncoupled, and their small-amplitude oscillations
can be studied with the asymptotic multi-scale expansion
\begin{equation}
        \left\{ \begin{array}{l} x_n(t) = \mu \left[ A_n(\mu^2 t) e^{i \omega t} + \bar{A}_n(\mu^2 t) e^{-i \omega t} \right] + \mu^3 X_n(t;\mu), \\[3pt]
y_n(t) = \mu \left[ B_n(\mu^2 t) e^{i \omega t} + \bar{B}_n(\mu^2 t) e^{-i \omega t} \right] + \mu^3 Y_n(t;\mu),
\end{array} \right.
\label{expansion}
\end{equation}
where $(A_n,B_n)$ are amplitudes for nearly harmonic oscillations and $(X_n,Y_n)$ are remainder terms.
In a similar context of single-array coupled nonlinear oscillators, it is shown in \cite{PPP}
how the asymptotic expansions like (\ref{expansion}) can be justified. From the conditions
that the remainder terms $(X_n,Y_n)$ remain bounded as the system evolves, the amplitudes $(A_n,B_n)$ are
shown to satisfy the discrete nonlinear Schr\"{o}dinger (dNLS) equations, which bring together all
the phenomena affecting the nearly harmonic oscillations (such as cubic nonlinear terms, the detuning frequency,
the coupling between the oscillators, and the amplitude of the parametric driving force).
A similar derivation for a single pair of coupled pendula is reported in \cite{barashenkovInt}.

Using the algorithm in \cite{PPP} and restricting the scopes of this derivation to the formal level,
we write the truncated system of equations for the remainder terms:
\begin{eqnarray}
\left\{ \begin{array}{l} \ddot{X}_n + X_n = F_n^{(1)} e^{i \omega t} + \overline{F_n^{(1)}} e^{-i \omega t}
+ F_n^{(3)} e^{3 i \omega t} + \overline{F_n^{(3)}} e^{-3 i \omega t} \\
 \ddot{Y}_n + Y_n = G_n^{(1)} e^{i \omega t} + \overline{G_n^{(1)}} e^{-i \omega t}
+ G_n^{(3)} e^{3 i \omega t} + \overline{G_n^{(3)}} e^{-3 i \omega t},
\end{array} \right. \quad n \in \mathbb{Z}, \quad t \in \mathbb{R},
\label{linear-equation}
\end{eqnarray}
where $F_n^{(1,3)}$ and $G_n^{(1,3)}$ are uniquely defined. Bounded solutions
to the linear inhomogeneous equations (\ref{linear-equation}) exist if and only if
$F_n^{(1)} = G_n^{(1)} = 0$ for every $n \in \mathbb{Z}$. Straightforward computations
show that the conditions $F_n^{(1)} = G_n^{(1)} = 0$ are equivalent to
the following evolution equations for slowly varying amplitudes $(A_n,B_n)$:
\begin{eqnarray}
\left\{ \begin{array}{l} 2 i \dot{A}_n = \epsilon \left( A_{n+1} - 2A_n + A_{n-1} \right) + \Omega A_n + \gamma \bar{B}_n + \frac{1}{2} |A_n|^2 A_n, \\
2 i \dot{B}_n = \epsilon \left( B_{n+1} - 2 B_n + B_{n-1} \right) + \Omega B_n + \gamma \bar{A}_n + \frac{1}{2} |B_n|^2 B_n,
\end{array} \right. \quad n \in \mathbb{Z}, \quad t \in \mathbb{R}.
\label{dNLS}
\end{eqnarray}
The system (\ref{dNLS}) takes the form of coupled parametrically forced dNLS equations (see \cite{parametricNLS}
for references on parametrically forced NLS equations).

There exists an invariant reduction of system (\ref{dNLS}) given by 
\begin{equation}
\label{synchr-1}
A_n = B_n, \quad n \in \mathbb{Z}
\end{equation}
to the scalar parametrically forced dNLS equation. 
The reduction (\ref{synchr-1}) corresponds to the symmetric synchronized oscillations of coupled pendula
of the model (\ref{oscillators}) with
\begin{equation}
\label{synchr-2}
x_n = y_n, \quad n \in \mathbb{Z}.
\end{equation}
In what follows, we consider a more general class of synchronized oscillations of coupled
pendula.

It turns out that the model (\ref{dNLS}) can be cast to the form
of the parity--time reversal ($\PT$) dNLS equations \cite{barashenkovInt}. Using the variables
\begin{equation}
\label{variables}
u_n := \frac{1}{4} \left( A_n - i \bar{B}_n \right), \quad v_n := \frac{1}{4} \left( A_n + i \bar{B}_n \right),
\end{equation}
the system of coupled dNLS equations (\ref{dNLS}) is rewritten in the equivalent
form
\begin{eqnarray}
\left\{ \begin{array}{l} 2 i \dot{u}_n = \epsilon \left( v_{n+1} - 2 v_n + v_{n-1} \right) + \Omega v_n + i \gamma u_n +
                2\left[ \left( 2|u_n|^2 + |v_n|^2 \right) v_n + u_n^2 \bar{v}_n \right], \\[3pt]
2 i \dot{v}_n = \epsilon \left( u_{n+1} - 2 u_n + u_{n-1} \right)  + \Omega u_n - i \gamma v_n +
2\left[ \left( |u_n|^2 + 2 |v_n|^2 \right) u_n + \bar{u}_n v_n^2 \right], \end{array} \right.
\label{PT-dNLS}
\end{eqnarray}
which is the starting point for our analytical and numerical work. The invariant reduction
(\ref{synchr-1}) for system (\ref{dNLS}) becomes
\begin{equation}
\label{reduction}
{\rm Im} (e^{\frac{i \pi}{4}} u_n) = 0, \quad {\rm Im} (e^{-\frac{i \pi}{4}} v_n) = 0, \quad n \in \mathbb{Z}.
\end{equation}

Without loss of generality, one can scale parameters $\Omega$, $\epsilon$, and $\gamma$
by a factor of two in order to eliminate the numerical factors in the system (\ref{PT-dNLS}). Also in the context of hard
nonlinear oscillators (e.g. in the framework of the $\phi^4$ theory), the cubic nonlinearity may have the opposite
sign compared to the one in the system (\ref{PT-dNLS}). However, given the applied context of the system
of coupled pendula, we will stick to the  specific form (\ref{PT-dNLS}) in further analysis.

\section{Symmetries and conserved quantities}
\label{symmetry}

The system of coupled dNLS equations (\ref{PT-dNLS}) is referred to as the $\PT$-symmetric dNLS equation
because the solutions remain invariant with respect to the action of the parity $\mathcal{P}$ and time-reversal 
$\mathcal{T}$ operators given by
\begin{equation}
\label{PT-symmetry}
\mathcal{P} \left[ \begin{array}{c} u \\ v \end{array} \right] =
\left[ \begin{array}{c} v \\ u \end{array} \right], \qquad
\mathcal{T} \left[ \begin{array}{c} u(t) \\ v(t) \end{array} \right] =
\left[ \begin{array}{c} \bar{u}(-t) \\ \bar{v}(-t) \end{array} \right].
\end{equation}
The parameter $\gamma$ introduces the gain--loss coefficient in each pair of coupled oscillators
due to the resonant periodic force. In the absence of all other effects,
the $\gamma$-term of the first equation of system (\ref{PT-dNLS}) induces the exponential growth of amplitude $u_n$,
whereas the $\gamma$-term of the second equation induces the exponential decay of amplitude $v_n$, if $\gamma > 0$.

The system (\ref{PT-dNLS}) truncated at a single site (say $n = 0$) is called the $\PT$-symmetric dimer.
In the work of Barashenkov {\em et al.} \cite{barashenkovInt}, it was shown
that all $\PT$-symmetric dimers with physically relevant cubic nonlinearities represent
Hamiltonian systems in appropriately introduced canonical variables. However, the
$\PT$-symmetric dNLS equation on a lattice does not typically have a Hamiltonian form if $\gamma \neq 0$.

Nevertheless, the particular nonlinear functions arising in the system (\ref{PT-dNLS}) correspond
to the $\PT$-symmetric dimers with a cross--gradient Hamiltonian structure \cite{barashenkovInt},
where variables $(u_n,\bar{v}_n)$ are canonically conjugate.
As a result, the system (\ref{PT-dNLS}) on the chain $\mathbb{Z}$
has additional conserved quantities. This fact looked like a mystery
in the recent works \cite{barashenkov,barashenkovInt}.

Here we clarify the mystery in the context of the derivation of the $\PT$-symmetric dNLS
equation (\ref{PT-dNLS}) from the original system (\ref{oscillators}).
Indeed, the system (\ref{oscillators}) of classical Newton particles
has a standard Hamiltonian structure with the energy function
\begin{eqnarray}
\nonumber
H_{x,y}(t) & = & \sum_{n \in \mathbb{Z}} \frac{1}{2} (\dot{x}_n^2 + \dot{y}_n^2) + 2 - \cos(x_n) - \cos(y_n) \\
& \phantom{t} & + \frac{1}{2} C (x_{n+1}-x_n)^2 + \frac{1}{2} C (y_{n+1}-y_n)^2 - D(t) x_n y_n.
\label{energy-oscillators}
\end{eqnarray}
Since the periodic movement of common strings for each pair of pendula
result in the time-periodic coefficient $D(t)$, the energy $H_{x,y}(t)$ is a periodic
function of time $t$. In addition, no other conserved quantities such as momenta exist
typically in lattice differential systems such as the system (\ref{oscillators})
due to broken continuous translational symmetry.

After the system (\ref{oscillators}) is reduced to the coupled dNLS equations (\ref{dNLS})
with the asymptotic expansion (\ref{expansion}), we can write the evolution problem (\ref{dNLS}) in the Hamiltonian form
with the standard straight-gradient symplectic structure
\begin{equation}
2 i \frac{d A_n}{d t} = \frac{\partial H_{A,B}}{\partial \bar{A}_n}, \quad
2 i \frac{d B_n}{d t} = \frac{\partial H_{A,B}}{\partial \bar{B}_n}, \quad n \in \mathbb{Z},
\label{Hamiltonian-dNLS}
\end{equation}
where the time variable $t$ stands now for the slow time $\mu^2 t$ and the energy function is
\begin{eqnarray}
\nonumber
H_{A,B} & = & \sum_{n \in \mathbb{Z}} \frac{1}{4} (|A_n|^4 + |B_n|^4) + \Omega (|A_n|^2 + |B_n|^2) + \gamma ( A_n B_n + \bar{A}_n \bar{B}_n) \\
 & \phantom{t} & - \epsilon |A_{n+1}-A_n|^2 - \epsilon |B_{n+1}-B_n|^2.
\label{energy-dNLS}
\end{eqnarray}
The energy function $H_{A,B}$ is conserved in the time evolution of the Hamiltonian system (\ref{Hamiltonian-dNLS}).
In addition, there exists another conserved quantity
\begin{equation}
\label{charge-dNLS}
Q_{A,B} = \sum_{n \in \mathbb{Z}} (|A_n|^2 - |B_n|^2),
\end{equation}
which is related to the gauge symmetry
$(A,B) \to (A e^{i \alpha},B e^{i \alpha})$ with $\alpha \in \mathbb{R}$
for solutions to the system (\ref{dNLS}).

When the transformation of variables (\ref{variables}) is used, the $\PT$-symmetric dNLS
equation (\ref{PT-dNLS}) is cast to the Hamiltonian form with the cross-gradient symplectic structure
\begin{equation}
2 i \frac{d u_n}{dt} = \frac{\partial H_{u,v}}{\partial \bar{v}_n}, \quad
2 i \frac{d v_n}{dt} = \frac{\partial H_{u,v}}{\partial \bar{u}_n}, \quad n \in \mathbb{Z},
\label{Hamiltonian-PT-dNLS}
\end{equation}
where the energy function is
\begin{eqnarray}
\nonumber
H_{u,v} & = & \sum_{n \in \mathbb{Z}} (|u_n|^2 + |v_n|^2)^2 + (u_n \bar{v}_n + \bar{u}_n v_n)^2 + \Omega (|u_n|^2 + |v_n|^2) \\
& \phantom{t} & - \epsilon |u_{n+1}-u_n|^2 - \epsilon |v_{n+1}-v_n|^2 + i \gamma ( u_n \bar{v}_n - \bar{u}_n v_n).
\label{energy-PT-dNLS}
\end{eqnarray}
The gauge-related function is written in the form
\begin{equation}
\label{charge-PT-dNLS}
Q_{u,v} = \sum_{n \in \mathbb{Z}} (u_n \bar{v}_n + \bar{u}_n v_n).
\end{equation}
The functions $H_{u,v}$ and $Q_{u,v}$ are conserved in the time evolution of the system (\ref{PT-dNLS}).
These functions follow from (\ref{energy-dNLS}) and (\ref{charge-dNLS}) after the transformation
(\ref{variables}) is used. Thus, the cross-gradient Hamiltonian structure of the $\PT$-symmetric 
dNLS equation (\ref{PT-dNLS}) is inherited from the Hamiltonian structure of the coupled oscillator model (\ref{oscillators}).

\section{Breathers (time-periodic solutions)}
\label{breathers}

We characterize the existence of breathers supported by the $\PT$-symmetric dNLS equation (\ref{PT-dNLS}).
In particular, breather solutions are continued for small values of coupling constant $\epsilon$
from solutions of the dimer equation arising at a single site, say the central site at $n = 0$.
We shall work in a sequence space $\ell^2(\mathbb{Z})$ of square integrable complex-valued sequences.

Time-periodic solutions to the $\PT$-symmetric dNLS equation (\ref{PT-dNLS})
are given in the form \cite{pel2,pel3}:
\begin{equation}
\label{stationary}
u(t) = U e^{-\frac{1}{2} i E t}, \quad v(t) = V e^{- \frac{1}{2} i E t},
\end{equation}
where the parameter $E$ is considered to be real, the factor $1/2$ is introduced for convenience, and
the sequence $(U,V)$ is time-independent. The breather (\ref{stationary}) is a localized mode
if $(U,V) \in \ell^2(\mathbb{Z})$, which implies that $|U_n|,|V_n| \to 0$
as $|n| \to \infty$. The breather (\ref{stationary}) is considered to be $\PT$-symmetric
with respect to the operators in (\ref{PT-symmetry}) if $V = \bar{U}$.

The reduction (\ref{reduction}) for symmetric synchronized oscillations
is satisfied if
\begin{equation}
\label{reduction-E}
E = 0 : \quad {\rm Im}( e^{\frac{i \pi}{4}} U_n) = 0, \quad {\rm Im} (e^{-\frac{i \pi}{4}} V_n) = 0,
\quad  n \in \mathbb{Z}.
\end{equation}
The time-periodic breathers (\ref{stationary}) with $E \neq 0$
generalize the class of symmetric synchronized oscillations (\ref{reduction-E}).

The time-independent sequence $(U,V) \in \ell^2(\mathbb{Z})$ can be found from the stationary
$\PT$-symmetric dNLS equation:
\begin{eqnarray}
\left\{ \begin{array}{l} E U_n = \epsilon \left( V_{n+1} - 2 V_n + V_{n-1} \right) + \Omega V_n + i \gamma U_n +
                2\left[ \left( 2|U_n|^2 + |V_n|^2 \right) V_n + U_n^2 \bar{V}_n \right], \\[3pt]
E V_n = \epsilon \left( U_{n+1} - 2 U_n + U_{n-1} \right) + \Omega U_n - i \gamma V_n +
2\left[ \left( |U_n|^2 + 2 |V_n|^2 \right) U_n + \bar{U}_n V_n^2 \right].
\end{array} \right.
\label{eq:statE}
\end{eqnarray}
The $\PT$-symmetric breathers with $V = \bar{U}$ satisfy the following scalar difference equation
\begin{equation}
E U_n  = \epsilon \left( \bar{U}_{n+1} - 2 \bar{U}_n + \bar{U}_{n-1} \right) + \Omega \bar{U}_n + i \gamma U_n +
6 |U_n|^2 \bar{U}_n + 2 U_n^3. \label{eq:statPT}
\end{equation}
Note that the reduction (\ref{reduction-E}) is compatible with equation (\ref{eq:statPT}) in the sense
that if $E = 0$ and $U_n = R_n e^{-i\pi/4}$, then $R$ satisfies a real-valued difference equation.

Let us set $\epsilon = 0$ for now and consider solutions to the dimer equation at the
central site $n = 0$:
\begin{equation}
\label{dimer-states}
(E - i \gamma) U_0 - \Omega \bar{U}_0 = 6 |U_0|^2 \bar{U}_0 + 2 U_0^3.
\end{equation}
The parameters $\gamma$ and $\Omega$ are considered to be fixed, and the breather parameter $E$ is thought
to parameterize continuous branches of solutions to the nonlinear algebraic equation (\ref{dimer-states}).
The solution branches depicted on Figure \ref{branches} are given in the following lemma.

\begin{figure}[t]
\center
\includegraphics[scale=0.6]{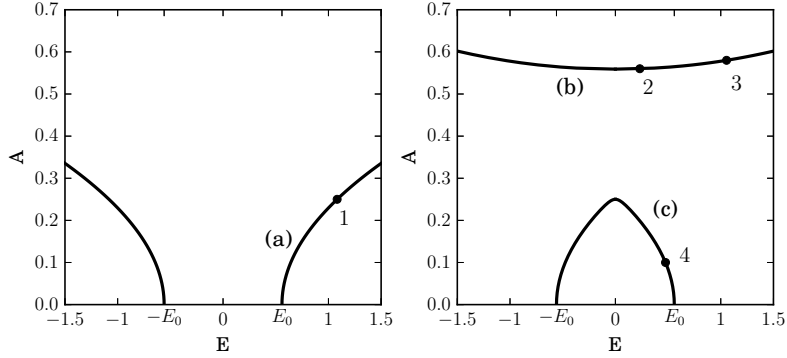}
\caption{Solution branches for the dimer equation (\ref{dimer-states}).}
\label{branches}
\end{figure}

\begin{lemma}
\label{lemma-dimer}
Assume $\gamma \neq 0$. The algebraic equation (\ref{dimer-states}) admits
the following solutions depending on $\gamma$ and $\Omega$:
\begin{itemize}
\item[(a)] $\Omega > |\gamma|$ - two symmetric unbounded branches exist for $\pm E > E_0$,
\item[(b)] $\Omega < |\gamma|$ - an unbounded branch exists for every $E \in \mathbb{R}$,
\item[(c)] $\Omega < -|\gamma|$ - a bounded branch exists for $-E_0 < E < E_0$,
\end{itemize}
where $E_0 := \sqrt{\Omega^2 - \gamma^2}$.
\end{lemma}

\begin{proof}
Substituting the decomposition $U_0 = A e^{i \theta}$ with $A > 0$ and $\theta \in [-\pi,\pi)$
into the algebraic equation (\ref{dimer-states}), we obtain
\begin{equation}
\label{parameterization-A}
\sin(2\theta) = \frac{\gamma}{4 A^2 + \Omega}, \quad \cos(2 \theta) = \frac{E}{8 A^2 + \Omega}.
\end{equation}
Excluding $\theta$ by using the fundamental trigonometric identity, we obtain the
explicit parametrization of the solutions to the algebraic equation (\ref{dimer-states})
by the amplitude parameter $A$:
\begin{equation}
\label{branch-E-A}
E^2 = (8 A^2 + \Omega)^2 \left[ 1 - \frac{\gamma^2}{(4 A^2+\Omega)^2} \right].
\end{equation}
The zero-amplitude limit $A = 0$ is reached if $|\Omega| > |\gamma|$,
in which case $E = \pm E_0$, where $E_0 := \sqrt{\Omega^2 - \gamma^2}$.
If $|\Omega| < |\gamma|$ , the solution branches (if they exist) are bounded away from
the zero solution.

Now we analyze the three cases of parameters $\gamma$ and $\Omega$ formulated in the lemma.

\begin{itemize}
\item[(a)] If $\Omega > |\gamma|$, then the parametrization (\ref{branch-E-A}) yields
a monotonically increasing map $\mathbb{R}^+ \ni A^2 \mapsto E^2 \in (E_0^2,\infty)$ because
\begin{equation}
\label{positivity-A}
\frac{d E^2}{d A^2} = \frac{8 (8A^2 + \Omega)}{(4A^2 + \Omega)^3}
\left[ 2 (4 A^2 + \Omega)^3 - \gamma^2 \Omega \right] > 0.
\end{equation}
In the two asymptotic limits, we obtain from (\ref{branch-E-A}):
$$
E^2 = E_0^2 + \mathcal{O}(A^2) \quad \mbox{\rm as} \quad A \to 0 \quad
\mbox{\rm and} \quad E^2 = 64 A^4 + \mathcal{O}(A^2) \quad \mbox{\rm as} \quad A \to \infty.
$$
See Figure~\ref{branches}(a).

\item[(b)] If $\Omega < |\gamma|$, the parametrization (\ref{branch-E-A}) yields
a monotonically increasing map $(A_+^2,\infty) \ni A^2 \mapsto E^2 \in \mathbb{R}^+$,
where
\begin{equation}
\label{A-plus}
A_+^2 := \frac{|\gamma| - \Omega}{4}.
\end{equation}
Indeed, we note that $4 A^2 + \Omega \geq 4A_+^2 + \Omega = |\gamma| > 0$ and
$$
2 (4 A^2 + \Omega)^3 - \gamma^2 \Omega \geq \gamma^2 (2 |\gamma| - \Omega) > 0,
$$
so that the derivative in  (\ref{positivity-A}) is positive for every $A^2 \geq A_+^2$.
We have
$$
E^2 \to 0 \quad \mbox{\rm as} \quad A^2 \to A_+^2 \quad {\rm and} \quad
E^2 = 64 A^4 + \mathcal{O}(A^2) \quad \mbox{\rm as} \quad A \to \infty.
$$
See Figure~\ref{branches}(b).

\item[(c)] If $\Omega < -|\gamma|$, then the parametrization (\ref{branch-E-A}) yields
a monotonically decreasing map $(0,A_-^2) \ni A^2 \mapsto E^2 \in (0,E_0^2)$,
where
\begin{equation}
\label{A-minus}
A_-^2 := \min\left\{\frac{|\Omega| - |\gamma|}{4},\frac{|\Omega|}{8}\right\}.
\end{equation}
In (\ref{A-minus}), the first choice is made if $|\Omega| \in (|\gamma|,2|\gamma|)$ and
the second choice is made if $|\Omega| \in (2|\gamma|,\infty)$. Both choices are the same
if $|\Omega| = 2 |\gamma|$. We note that
$8 A^2 \leq |\Omega|$, therefore, the derivative (\ref{positivity-A}) needs to be
rewritten in the form
\begin{equation}
\label{negativity-A}
\frac{d E^2}{d A^2} = -\frac{8 (|\Omega| - 8A^2)}{(|\Omega|  - 4 A^2)^3}
\left[ 2 (|\Omega| - 4A^2)^3 - \gamma^2 |\Omega| \right] < 0,
\end{equation}
where $2 (|\Omega| - 4A^2)^3 - \gamma^2 |\Omega| > 0$
for both $|\Omega| \in (|\gamma|,2|\gamma|)$ and $|\Omega| \in [2|\gamma|,\infty)$.
In the two asymptotic limits, we obtain from (\ref{branch-E-A}):
$$
E^2 = E_0^2 + \mathcal{O}(A^2) \quad \mbox{\rm as} \quad A \to 0 \quad \mbox{\rm and}
\quad E^2 \to 0 \quad \mbox{\rm as} \quad A^2 \to A_-^2.
$$
See Figure~\ref{branches}(c).
\end{itemize}
Note that branches (b) and (c) coexist for $\Omega < -|\gamma|$.
\end{proof}

\begin{remark}
The reduction (\ref{reduction-E}) corresponds to the choice:
$$
E = 0, \quad \theta = -\frac{\pi}{4}, \quad  4 A^2 + \Omega + \gamma = 0.
$$
If $\gamma > 0$, this choice corresponds to $A = A_-$ for $\Omega \in (-2|\gamma|,-|\gamma|)$,
that is, the point $E = 0$ on branch (c). If $\gamma < 0$, it corresponds to $A = A_+$ for any $\Omega < |\gamma|$,
that is, the point $E = 0$ on branch (b).
\end{remark}

Every solution of Lemma \ref{lemma-dimer} can be extended to a breather on the chain $\mathbb{Z}$
which satisfies the spatial symmetry condition in addition to the $\PT$ symmetry:
\begin{equation}
\label{spatial-symmetry}
U_{-n} = U_n = \bar{V}_n = \bar{V}_{-n}, \quad n \in \mathbb{Z}.
\end{equation}
In order to prove the existence of the symmetric breather solution
to the difference equation (\ref{eq:statPT}), we use the following implicit function theorem.\\

{\bf Implicit Function Theorem} (Theorem 4.E in \cite{Zeidler}). {\em
Let $X,Y$ and $Z$ be Banach spaces and let $F(x,y)\colon X\times Y \to Z$
be a $C^1$ map on an open neighborhood of the point $(x_0,y_0)\in X\times Y$.
Assume that
$$
F(x_0,y_0) = 0
$$
and that
$$
D_x F(x_0,y_0) \colon X\to Z \text{ is one-to-one and onto. }
$$
There are $r>0$ and $\sigma>0$ such that for each $y$ with
$\| y - y_0 \|_Y \le \sigma$ there exists a unique solution
$x\in X$ of the nonlinear equation $F(x,y) = 0$ with
$\|x - x_0\|_X \le r$. Moreover, the map $Y\owns y\mapsto x \in X$
is $C^1$ near $y = y_0$.
}

\vspace{0.25cm}

With two applications of the implicit function theorem,
we prove the following main result of this section.

\begin{theorem}
\label{theorem-existence-soliton}
Fix $\gamma \neq 0$, $\Omega \neq -2 |\gamma|$, and $E \neq \pm E_0$,
where $E_0 := \sqrt{\Omega^2 - \gamma^2} > 0$ if $|\Omega| > |\gamma|$.
There exists $\epsilon_0 > 0$ sufficiently small and $C_0 > 0$ such that
for every $\epsilon \in (-\epsilon_0,\epsilon_0)$, there exists a unique
solution $U \in l^2(\mathbb{Z})$ to the difference equation (\ref{eq:statPT})
satisfying the symmetry (\ref{spatial-symmetry}) and the bound
\begin{equation}
\label{bound-soliton}
\left| U_0 - A e^{i \theta} \right| + \sup_{n \in \mathbb{N}} |U_n| \leq C_0 |\epsilon|,
\end{equation}
where $A$ and $\theta$ are defined in Lemma \ref{lemma-dimer}. Moreover, the solution $U$
is smooth in $\epsilon$.
\end{theorem}

\begin{proof}
In the first application of the implicit function theorem, we consider the following system of algebraic equations
\begin{equation}
E U_n = \epsilon \left( \bar{U}_{n+1} - 2 \bar{U}_n + \bar{U}_{n-1} \right) + \Omega \bar{U}_n + i \gamma U_n +
6 |U_n|^2 \bar{U}_n + 2 U_n^3, \quad n \in \mathbb{N}, \label{fixed-point-1}
\end{equation}
where $U_0 \in \mathbb{C}$ is given, in addition to parameters $\gamma$, $\Omega$, and $E$.

Let $x = \{ U_n \}_{n \in \mathbb{N}}$, $X = \ell^2(\mathbb{N})$, $y = \epsilon$, $Y = \mathbb{R}$,
and $Z = \ell^2(\mathbb{N})$. Then, we have $F(0,0) = 0$ and the Jacobian
operator $D_x F(0,0)$ is given by identical copies of the matrix
$$
\left[ \begin{matrix} E - i \gamma & -\Omega \\ -\Omega & E + i \gamma \end{matrix} \right],
$$
with the eigenvalues $\lambda_{\pm} := E \pm \sqrt{\Omega^2 - \gamma^2}$. By the assumption of the lemma,
$\lambda_{\pm} \neq 0$, so that the Jacobian operator $D_x F(0,0)$ is one-to-one and onto.
By the implicit function theorem, for every $U_0 \in \mathbb{C}$
and every $\epsilon \neq 0$ sufficiently small, there exists a unique small
solution $U \in \ell^2(\mathbb{N})$ of the system (\ref{fixed-point-1}) such that
\begin{equation}
\| U \|_{l^2(\mathbb{N})} \leq C_1 |\epsilon| |U_0|,
\label{eq:U}
\end{equation}
where the positive constant $C_1$ is independent from $\epsilon$ and $U_0$.

Thanks to the symmetry of the difference equation (\ref{eq:statPT}),
we find that $U_{-n} = U_n$, $n \in \mathbb{N}$ satisfy the same system (\ref{fixed-point-1})
with $-n \in \mathbb{N}$, with the same unique solution.

In the second application of the implicit function theorem, we consider the
following algebraic equation
\begin{equation}
E U_0 = 2 \epsilon \left( \bar{U}_{1} - \bar{U}_0 \right) + \Omega \bar{U}_0 + i \gamma U_0 +
6 |U_0|^2 \bar{U}_0 + 2 U_0^3, \label{fixed-point-2}
\end{equation}
where $U_1 \in \mathbb{C}$ depends on $U_0$, $\gamma$, $\Omega$, and $E$,
satisfies the bound (\ref{eq:U}), and is uniquely defined by the previous result.

Let $x = U_0$, $X = \mathbb{C}$, $y = \epsilon$, $Y = \mathbb{R}$,
and $Z = \mathbb{C}$. Then, we have $F(A e^{i\theta},0) = 0$, where $A$ and $\theta$ are
defined in Lemma \ref{lemma-dimer}. The Jacobian
operator $D_x F(A e^{i \theta},0)$ is given by the matrix
\begin{eqnarray}
\nonumber
& \phantom{t} & \left[ \begin{matrix} E - i \gamma - 6 U_0^2 - 6 \bar{U}_0^2 & -\Omega - 12 |U_0|^2 \\ -\Omega - 12 |U_0|^2 &
E + i \gamma - 6 U_0^2 - 6 \bar{U}_0^2 \end{matrix} \right] \biggr|_{U_0 = A e^{i \theta}} \\
\label{matrix-Jacobian}
& = & \left[ \begin{array}{cc} E - i \gamma - \frac{12 E A^2}{\Omega + 8 A^2} & -\Omega - 12 A^2 \\
-\Omega - 12 A^2 & E + i \gamma - \frac{12 E A^2}{\Omega + 8 A^2} \end{array} \right].
\end{eqnarray}
We show in Lemma \ref{proposition-unbounded} below that
the matrix given by (\ref{matrix-Jacobian}) is invertible under the conditions
$\gamma \neq 0$ and $\Omega \neq - 2 |\gamma|$. By the implicit function theorem,
for every $\epsilon \neq 0$ sufficiently small, there exists a unique solution $U_0 \in \mathbb{C}$
to the algebraic equation (\ref{fixed-point-2}) near $A e^{i \theta}$ such that
\begin{equation}
\left| U_0 - A e^{i \theta} \right| \leq C_2 |\epsilon|,
\label{eq:U0}
\end{equation}
where the positive constant $C_2$ is independent from $\epsilon$.
The bound (\ref{bound-soliton}) holds thanks to the bounds (\ref{eq:U}) and (\ref{eq:U0}).
Since both equations (\ref{fixed-point-1}) and (\ref{fixed-point-2}) are smooth in $\epsilon$,
the solution $U$ is smooth in $\epsilon$.
\end{proof}

In the following result, we show that the matrix given by (\ref{matrix-Jacobian}) is invertible
for every branch of Lemma \ref{lemma-dimer} with an exception of a single point $E = 0$
on branch (c) for $\Omega = -2 |\gamma|$.

\begin{lemma}
\label{proposition-unbounded}
With the exception of the point $E = 0$ on branch (c) of Lemma \ref{lemma-dimer} for $\Omega = -2 |\gamma|$,
the matrix given by (\ref{matrix-Jacobian}) is invertible for every $\gamma \neq 0$.
\end{lemma}

\begin{proof}
The matrix given by (\ref{matrix-Jacobian}) has zero eigenvalue if and only if
its determinant is zero, which happens at
$$
\frac{E^2 (\Omega - 4 A^2)^2}{(\Omega + 8 A^2)^2} + \gamma^2 - (\Omega + 12 A^2)^2 = 0.
$$
Eliminating $E^2$ by using parametrization (\ref{branch-E-A}) and simplifying the algebraic equation
for nonzero $A^2$, we reduce it to the form
\begin{equation}
2(\Omega + 4 A^2)^3 = \Omega \gamma^2.
\label{constraint-parameters}
\end{equation}
We now check if this constraint can be satisfied for the three branches of Lemma \ref{lemma-dimer}.
\begin{itemize}
\item[(a)] If $\Omega > |\gamma|$, the constraint (\ref{constraint-parameters}) is not satisfied because
the left-hand side
$$
2(\Omega + 4 A^2)^3 \geq 2\Omega^3 > 2 \Omega \gamma^2
$$
exceeds the right-hand side $\Omega \gamma^2$.

\item[(b)] If $\Omega < |\gamma|$ and $A^2 \geq A_+^2$,
where $A_+^2$ is given by (\ref{A-plus}), the constraint (\ref{constraint-parameters}) is not satisfied
because the left-hand side
$$
2(\Omega + 4 A^2)^3 \geq 2(\Omega + 4 A_+^2)^3 = 2 |\gamma|^3
$$
exceeds the left-hand side $\Omega \gamma^2$ both for $\Omega \in [0,|\gamma|)$ and for $\Omega < 0$.

\item[(c)] If $\Omega < -|\gamma|$ and $A^2 \leq A_-^2$, where $A_-^2$ is given by (\ref{A-minus}),
the constraint (\ref{constraint-parameters}) is not satisfied because the left-hand side is estimated by
$$
2(4 A^2 + \Omega)^3  \leq 2(4 A_-^2 - |\Omega|)^3 = \min\{ -2 |\gamma|^3, -|\Omega|^3/4\}.
$$
In the first case, we have $|\Omega| \in (|\gamma|,2 |\gamma|)$, so that the left-hand side is strictly smaller than
$-|\Omega| \gamma^2$. In the second case, we have $|\Omega| > 2 |\gamma|$, so that
the left-hand side is also strictly smaller
than $-|\Omega| \gamma^2$. Only if $|\Omega| = 2 |\gamma|$, the constraint (\ref{constraint-parameters})
is satisfied at $E = 0$, when $A^2 = A_-^2$ and
$$
2(4 A^2 + \Omega)^3 = -2 |\gamma|^3 = -|\Omega| \gamma^2 = \Omega \gamma^2.
$$
\end{itemize}
Hence, the matrix (\ref{matrix-Jacobian}) is invertible for all parameter values with one exceptional case.
\end{proof}

\begin{remark}
In the asymptotic limit $E^2 = 64 A^4 + \mathcal{O}(A^2)$ as $A \to \infty$, see Lemma \ref{lemma-dimer},
the matrix (\ref{matrix-Jacobian}) is expanded asymptotically as
\begin{eqnarray}
-\frac{1}{2} \left[ \begin{array}{cc} E & 3 |E| \\
3 |E| & E \end{array} \right] + \mathcal{O}(1) \quad \mbox{\rm as} \quad |E| \to \infty,
\label{matrix-Jacobian-asymptotics}
\end{eqnarray}
with the two eigenvalues $\lambda_1 = E$ and $\lambda_2 = -2E$. Thus, the matrix given by (\ref{matrix-Jacobian-asymptotics})
is invertible for every branch extending to sufficiently large values of $E$.
\end{remark}

\section{Stability of zero equilibrium}
\label{zero-equilibrium}

Here we discuss the linear stability of the zero equilibrium in the $\PT$-symmetric dNLS equation (\ref{PT-dNLS}). 
The following proposition yields a simple result. 

\begin{proposition}
The zero equilibrium of the $\PT$-symmetric dNLS equation (\ref{PT-dNLS})
is linearly stable if $|\gamma| < \gamma_0$, where
\begin{equation}
\label{gamma-0}
\gamma_0 := \left\{ \begin{array}{ll} \Omega - 4 \epsilon, & \Omega > 0, \\
|\Omega|, & \Omega < 0. \end{array} \right.
\end{equation}
The zero equilibrium is linearly unstable if $|\gamma| > \gamma_0$. 
\label{prop-zero-equilibrium}
\end{proposition}

\begin{proof}
Truncating the $\PT$-symmetric dNLS equation (\ref{PT-dNLS}) at the linear terms
and using the Fourier transform
\begin{equation}
\label{discrete-Fourier}
u_n(t) = \frac{1}{2\pi} \int_{-\pi}^{\pi} \hat{U}(k) e^{i k n + i \omega(k) t} dk,
\end{equation}
we obtain the linear homogeneous system
$$
\hat{D}(k) \left[ \begin{array}{c} \hat{U}(k) \\ \hat{V}(k) \end{array} \right] = \left[ \begin{array}{c} 0 \\ 0 \end{array} \right],
\quad \mbox{\rm where} \quad
\hat{D}(k) := \left[ \begin{matrix} -2 \omega(k) - i \gamma & -\Omega + 4 \epsilon \sin^2(k/2) \\
-\Omega + 4 \epsilon \sin^2(k/2) & -2 \omega(k) + i \gamma \end{matrix} \right].
$$
The determinant of $\hat{D}(k)$ is zero if and only if $\omega(k)$ is found from the quadratic equation
\begin{equation}
\label{dispersion-relation}
4 \omega^2(k) + \gamma^2 - \left( \Omega - 4 \epsilon \sin^2\frac{k}{2} \right)^2 = 0.
\end{equation}
For any $|\gamma| < \gamma_0$, where $\gamma_0$ is given by (\ref{gamma-0}),
the two branches $\pm \omega(k)$ found from the quadratic equation (\ref{dispersion-relation})
are real-valued and non-degenerate for every $k \in [-\pi,\pi]$.
Therefore, the zero equilibrium is linearly stable.

On the other hand, for any $|\gamma| > \gamma_0$, the values of $\omega(k)$ are purely imaginary either near
$k = \pm \pi$ if $\Omega > 0$ or near $k = 0$ if $\Omega < 0$. Therefore, the zero equilibrium is linearly unstable.
\end{proof}

\begin{remark}
The value $\gamma_0$ given by (\ref{gamma-0}) represents the phase transition threshold
and the $\PT$-symmetric dNLS equation (\ref{PT-dNLS}) is said to have broken
$\PT$-symmetry for $|\gamma| > \gamma_0$.
\end{remark}

If $\epsilon = 0$, the zero equilibrium is only linearly stable for $|\gamma| < |\Omega|$.
Since the localized breathers cannot be stable when the zero background is unstable, we shall study
stability of breathers only for the case when $|\gamma| < |\Omega|$, that is,
in the regime of unbroken $\PT$-symmetry.

\section{Variational characterization of breathers}
\label{variational}

It follows from Theorem \ref{theorem-existence-soliton} that
each interior point on the solution branches shown on Figure \ref{branches}
generates a fundamental breather of the $\PT$-symmetric dNLS equation (\ref{PT-dNLS}).
We shall now characterize these breathers as relative equilibria of the energy function.

Thanks to the cross-gradient symplectic structure (\ref{Hamiltonian-PT-dNLS}),
the stationary $\PT$-symmetric dNLS equation \eqref{eq:statE} can be written in the gradient form
\begin{equation}
E U_n = \frac{\partial H_{u,v}}{\partial \bar{V}_n}, \quad
E V_n = \frac{\partial H_{u,v}}{\partial \bar{U}_n}, \quad n \in \mathbb{Z}.
\label{Hamiltonian-PT-dNLS-var}
\end{equation}
Keeping in mind the additional conserved quantity $Q_{u,v}$ given by
(\ref{charge-PT-dNLS}), we conclude that the stationary solution $(U,V)$
is a critical point of the combined energy function given by
\begin{equation}
\label{combined-energy-functional}
H_E := H_{u,v} - E Q_{u,v}.
\end{equation}

If we want to apply the Lyapunov method in order to study nonlinear stability of stationary solutions
in Hamiltonian systems, we shall investigate convexity of the second variation
of the combined energy functional $H_E$ at $(U,V)$. Using the expansion
$u = U + {\bf u}$, $v = V + {\bf v}$ and introducing extended variables $\Phi$ and $\phi$ with the blocks
\begin{equation}
\label{notation-pert}
\Phi_n := (U_n,\bar{U}_n,V_n,\bar{V}_n), \quad
\phi_n := ({\bf u}_n,\bar{\bf u}_n,{\bf v}_n,\bar{\bf v}_n),
\end{equation}
we can expand the smooth function $H_E$ up to the quadratic terms in $\phi$:
\begin{equation}
\label{quadratic-expansion}
H_E(\Phi+\phi) = H_E(\Phi) + \frac{1}{2} \left\langle \mathcal{H}''_E \phi, \phi \right\rangle_{\ltz} +
\mathcal{O}(\|\phi\|^3_{\ltz}),
\end{equation}
where $\mathcal{H}''_E$ is the self-adjoint (Hessian) operator defined on $\ell^2(\mathbb{Z})$ and
the scalar product was used in the following form:
$$
\langle x,y \rangle_{\ltz} = \sum_{k \in \mathbb{Z}} x_k \bar{y}_k.
$$

Using (\ref{energy-PT-dNLS}) and (\ref{charge-PT-dNLS}), the Hessian operator can be computed explicitly as follows
\begin{eqnarray}
\label{PT-dNLS-var}
\mathcal{H}''_E = \mathcal{L} + \epsilon \Delta,
\end{eqnarray}
where blocks of $\mathcal{L}$ at each lattice node $n \in \mathbb{Z}$ are given by
{\small \begin{eqnarray*}
\arraycolsep=0pt\def\arraystretch{1}
\mathcal{L}_n = \left[\begin{array}{cccc} \Omega + 8 |U_n|^2 & 2(U_n^2 + \bar{U}_n^2) & -E - i \gamma + 4(U_n^2 + \bar{U}_n^2) & 4 |U_n|^2 \\
2(U_n^2 + \bar{U}_n^2) & \Omega + 8 |U_n|^2 & 4 |U_n|^2  & -E + i \gamma + 4(U_n^2 + \bar{U}_n^2) \\
-E + i \gamma + 4(U_n^2 + \bar{U}_n^2) & 4 |U_n|^2 & \Omega + 8 |U_n|^2 & 2(U_n^2 + \bar{U}_n^2)\\
4 |U_n|^2 & -E - i \gamma + 4(U_n^2 + \bar{U}_n^2) & 2(U_n^2 + \bar{U}_n^2) & \Omega + 8 |U_n|^2 \end{array} \right]
\end{eqnarray*}}and $\Delta$ is the discrete Laplacian operator applied to blocks of $\phi$ at each lattice node $n \in \mathbb{Z}$:
$$
(\Delta \phi)_n = \phi_{n+1} - 2 \phi_n + \phi_{n-1}.
$$
In the expression for $\mathcal{L}_n$, we have used the $\PT$-symmetry condition $V = \bar{U}$ for
the given stationary solution $(U,V)$.

We study convexity of the combined energy functional $H_E$ at $(U,V)$. Since the zero equilibrium
is linearly stable only for $|\gamma| < |\Omega|$ (if $\epsilon = 0$),
we only consider breathers of Theorem \ref{theorem-existence-soliton} for $|\gamma| < |\Omega|$.
In order to study eigenvalues of $\mathcal{H}''_E$ for small values of $\epsilon$, we
use the following perturbation theory. \\

{\bf Perturbation Theory for Linear Operators} (Theorem VII.1.7 in \cite{Kato}). {\em
Let $T(\epsilon)$ be a family of bounded operators from Banach space $X$ to itself,
which depends analytically on the small parameter $\epsilon$. If the spectrum of $T(0)$
is separated into two parts, the subspaces of $X$ corresponding to the separated parts also depend
analytically on $\epsilon$. In particular, the spectrum of $T(\epsilon)$ is separated into two parts for any $\epsilon \neq 0$
sufficiently small. }

\vspace{0.25cm}

With an application of the perturbation theory for linear operators, we prove the following main result of this section.

\begin{theorem}
\label{theorem-critical-point}
Fix $\gamma \neq 0$, $\Omega$, and $E$ along branches of the $\PT$-symmetric breathers $(U,V)$ given by
Theorem  \ref{theorem-existence-soliton} such that $|\Omega| > |\gamma|$ and $E \neq \pm E_0$,
where $E_0 := \sqrt{\Omega^2 - \gamma^2} > 0$. For every $\epsilon > 0$ sufficiently small,
the operator $\mathcal{H}''_E$ admits a one-dimensional kernel in $\ell^2(\mathbb{Z})$ spanned
by the eigenvector $\sigma \Phi$ due to the gauge invariance, where the blocks of the eigenvector are given by
\begin{equation}
\label{kernel}
(\sigma \Phi)_n := (U_n, - \bar{U}_n, V_n, -\bar{V}_n).
\end{equation}
In addition,
\begin{itemize}
\item If $|E| > E_0$, the spectrum of $\mathcal{H}''_E$ in $\ell^2(\mathbb{Z})$
includes infinite-dimensional positive and negative parts.
\item If $|E| < E_0$ and $\Omega < -|\gamma|$, the spectrum of $\mathcal{H}\emph{}''_E$ in $\ell^2(\mathbb{Z})$
includes an infinite-dimensional negative part and either three or one simple positive eigenvalues
for branches (b) and (c) of Lemma \ref{lemma-dimer} respectively.
\end{itemize}
\end{theorem}

\begin{proof}
If $\epsilon = 0$, the breather solution of Theorem \ref{theorem-existence-soliton}
is given by $U_n = 0$ for every $n \neq 0$ and $U_0 = A e^{i \theta}$,
where $A$ and $\theta$ are defined by Lemma \ref{lemma-dimer}.
In this case, the linear operator $\mathcal{H}''_E = \mathcal{L}$ decouples
into $4$-by-$4$ blocks for each lattice node $n \in \mathbb{Z}$.

For $n = 0$, the $4$-by-$4$ block of the linear operator $\mathcal{L}$
is given by
{\small \begin{equation*} \label{block-L-0}
\arraycolsep=0pt\def\arraystretch{1}
\mathcal{L}_0 = \left[\begin{array}{cccc} \Omega + 8 A^2 & 4 A^2 \cos(2 \theta) & -E - i \gamma + 8 A^2 \cos(2 \theta) & 4 A^2 \\
4 A^2 \cos(2 \theta) & \Omega + 8 A^2 & 4 A^2  & -E + i \gamma + 8 A^2 \cos(2 \theta) \\
-E + i \gamma + 8 A^2 \cos(2 \theta) & 4 A^2 & \Omega + 8 A^2 & 4 A^2 \cos(2 \theta) \\
4 A^2 & -E - i \gamma + 8 A^2 \cos(2 \theta) & 4 A^2 \cos(2 \theta) & \Omega + 8 A^2 \end{array} \right].
\end{equation*}}Using relations (\ref{parameterization-A}) and (\ref{branch-E-A}), as well as
symbolic computations with MAPLE, we found that the $4$-by-$4$ matrix block $\mathcal{L}_0$ admits
a simple zero eigenvalue and three nonzero eigenvalues $\mu_1$, $\mu_2$, and $\mu_3$ given by
\begin{eqnarray}
\label{eigenvalue-1}
\mu_1 & = & 2 (4A^2 + \Omega), \\
\label{eigenvalue-2-3}
\mu_{2,3} & = & 12 A^2 + \Omega \pm \sqrt{(4A^2 - \Omega)^2 + \frac{16 \Omega A^2 \gamma^2}{(4A^2 + \Omega)^2}}.
\end{eqnarray}
For each branch of Lemma \ref{lemma-dimer} with $\gamma \neq 0$ and $E \neq \pm E_0$, we have $4A^2 + \Omega \neq 0$,
so that $\mu_1 \neq 0$. Furthermore, either $\mu_2 = 0$ or $\mu_3 = 0$ if and only if
$$
(12 A^2 + \Omega)^2 (4 A^2 + \Omega)^2 = (16 A^4 - \Omega^2)^2 + 16 \Omega \gamma^2 A^2.
$$
Expanding this equation for nonzero $A$ yields constraint (\ref{constraint-parameters}).
With the exception of a single point $E = 0$ at $\Omega = -2 |\gamma|$,
we showed in Lemma \ref{proposition-unbounded} that the constraint (\ref{constraint-parameters})
does not hold for any of the branches of Lemma \ref{lemma-dimer}. Therefore, $\mu_2 \neq 0$
and $\mu_3 \neq 0$ along each branch of Lemma \ref{lemma-dimer}
and the signs of $\mu_1$, $\mu_2$, and $\mu_3$ for each branch of Lemma \ref{lemma-dimer}
can be obtained in the limit $A \to \infty$ for branches (a) and (b) or $A \to 0$ for branch (c).
By means of these asymptotic computations as $A \to \infty$ or $A \to 0$, we obtain the following
results for the three branches shown on Figure \ref{branches}:
\begin{itemize}
\item[(a)] $\mu_1, \mu_2, \mu_3 > 0$.

\item[(b)] $\mu_1, \mu_2, \mu_3 > 0$.

\item[(c)] $\mu_1 < 0$, $\mu_2 > 0$, and $\mu_3 < 0$.
\end{itemize}

For $n \in \mathbb{Z} \backslash \{0\}$, the $4$-by-$4$ block of the linear operator $\mathcal{L}$
is given by
\begin{equation}
\label{block-L-n}
\mathcal{L}_n = \left[\begin{array}{cccc} \Omega & 0 & -E - i \gamma  & 0 \\
0 & \Omega & 0  & -E + i \gamma  \\
-E + i \gamma  & 0 & \Omega  & 0 \\
0 & -E - i \gamma & 0 & \Omega \end{array} \right].
\end{equation}
Each block has two double eigenvalues $\mu_+$ and $\mu_-$ given by
$$
\mu_+ = \Omega + \sqrt{E^2 + \gamma^2}, \quad \mu_- = \Omega - \sqrt{E^2 + \gamma^2}.
$$
Since there are infinitely many nodes with $n \neq 0$, the points $\mu_+$ and $\mu_-$ have infinite multiplicity
in the spectrum of the linear operator $\mathcal{L}$. Furthermore, we can sort up the signs of $\mu_+$ and $\mu_-$
for each point on the three branches shown on Figure \ref{branches}:
\begin{itemize}
\item[(1),(3)] If $|E| > E_0 := \sqrt{\Omega^2 - \gamma^2}$, then $\mu_+ > 0$ and $\mu_- < 0$.
\item[(2),(4)] If $|E| < E_0$ and $\Omega < -|\gamma|$, then $\mu_+, \mu_- < 0$.
\end{itemize}

By using the perturbation theory for linear operators, we argue as follows:
\begin{itemize}
\item Since $\mathcal{H}''_E$ is Hermitian on $\ell^2(\mathbb{Z})$, its spectrum is a subset of the real line
for every $\epsilon \neq 0$.

\item The zero eigenvalue persists with respect to $\epsilon \neq 0$ at zero
because the eigenvector (\ref{kernel}) belongs to the kernel of $\mathcal{H}''_E$
due to the gauge invariance for every $\epsilon \neq 0$.

\item The other eigenvalues of $\mathcal{L}$ are isolated away from zero. The spectrum of $\mathcal{H}''_E$ is continuous
with respect to $\epsilon$ and includes infinite-dimensional parts near points $\mu_+$ and $\mu_-$
for small $\epsilon > 0$ (which may include continuous spectrum and isolated eigenvalues) as well
as simple eigenvalues near $\mu_{1,2,3}$ (if $\mu_{1,2,3}$ are different from $\mu_{\pm}$).
\end{itemize}

The statement of the theorem follows from the perturbation theory and
the count of signs of $\mu_{1,2,3}$ and $\mu_{\pm}$ above.
\end{proof}

\begin{remark}
In the asymptotic limit $E^2 = 64 A^4 + \mathcal{O}(A^2)$ as $A \to \infty$,
we can sort out eigenvalues of $\mathcal{H}''_E$ asymptotically as:
\begin{eqnarray}
\mu_1 \approx |E|, \quad \mu_2 \approx 2 |E|, \quad \mu_3 \approx |E|, \quad \mu_+ \approx |E|, \quad \mu_- \approx -|E|,
\end{eqnarray}
where the remainder terms are $\mathcal{O}(1)$ as $|E| \to \infty$. The values
$\mu_1$, $\mu_3$, and $\mu_+$ are close to each other as $|E| \to \infty$.
\end{remark}

\begin{remark}
It follows from Theorem \ref{theorem-critical-point} that for $|E| > E_0$,
the breather $(U,V)$ is a saddle point of the energy functional $H_E$
with infinite-dimensional positive and negative invariant subspaces of
the Hessian operator $\mathcal{H}''_E$. This is very similar to the Hamiltonian
systems of the Dirac type, where stationary states are located in the gap
between the positive and negative continuous spectrum. This property holds 
for points 1 and 3 on branches (a) and (b) shown on Figure \ref{branches}.
\end{remark}

\begin{remark}
No branches other than $|E| > E_0$ exist for $\Omega > |\gamma|$.
On the other hand, points 2 and 4 on branches (b) and (c) shown on 
Figure \ref{branches} satisfy $|E| < E_0$ and $\Omega < -|\gamma|$. The breather
$(U,V)$ is a saddle point of $H_E$ for these points 
and it only has three (one) directions of positive energy in space $\ell^2(\mathbb{Z})$ for
point 2 (point 4).
\end{remark}

\section{Spectral and orbital stability of breathers}
\label{stability}

Spectral stability of breathers can be studied for small values
of coupling constant $\epsilon$ by using the perturbation theory \cite{pel3}. First, we linearize
the $\PT$-symmetric dNLS equation (\ref{PT-dNLS}) at the breather (\ref{stationary})
by using the expansion
$$
u(t) = e^{-\frac{1}{2} i Et} \left[ U + {\bf u}(t) \right], \quad
v(t) = e^{-\frac{1}{2} i Et} \left[ V + {\bf v}(t) \right],
$$
where $({\bf u},{\bf v})$ is a small perturbation satisfying the linearized equations
\begin{eqnarray}
\left\{ \begin{array}{l} 2 i \dot{{\bf u}}_n + E {\bf u}_n = \epsilon \left( {\bf v}_{n+1} - 2 {\bf v}_n + {\bf v}_{n-1} \right)
 + \Omega {\bf v}_n + i \gamma {\bf u}_n \\
\phantom{text} + 2 \left[ 2 \left( |U_n|^2 + |V_n|^2 \right) {\bf v}_n + (U_n^2 + V_n^2) \bar{{\bf v}}_n +
2 (\bar{U}_n V_n + U_n \bar{V}_n) {\bf u}_n + 2 U_n V_n \bar{{\bf u}}_n \right], \\
2 i \dot{{\bf v}}_n + E {\bf v}_n = \epsilon \left( {\bf u}_{n+1} - 2 {\bf u}_n + {\bf u}_{n-1} \right)
 + \Omega {\bf u}_n - i \gamma {\bf v}_n \\
\phantom{text} + 2 \left[ 2 \left( |U_n|^2 + |V_n|^2 \right) {\bf u}_n + (U_n^2 + V_n^2) \bar{{\bf u}}_n +
2 (\bar{U}_n V_n + U_n \bar{V}_n) {\bf v}_n + 2 U_n V_n \bar{{\bf v}}_n \right].
\end{array} \right.
\label{PT-dNLS-lin}
\end{eqnarray}

The spectral stability problem arises from the linearized equations (\ref{PT-dNLS-lin}) after the
separation of variables:
$$
{\bf u}(t) = \varphi e^{\frac{1}{2} \lambda t}, \quad
\bar{\bf u}(t) = \psi e^{\frac{1}{2} \lambda t}, \quad
{\bf v}(t) = \chi e^{\frac{1}{2} \lambda t}, \quad
\bar{\bf v}(t) = \nu e^{\frac{1}{2} \lambda t}.
$$
where $\phi := (\varphi,\psi,\chi,\nu)$ is the eigenvector corresponding to the spectral parameter $\lambda$.
Note that $(\varphi,\psi)$ and $(\chi,\nu)$ are no longer
complex conjugate to each other if $\lambda$ has a nonzero imaginary part.
The spectral problem can be written in the explicit form
\begin{equation}
\label{spectral-problem}
\left\{ \begin{array}{l}  (E + i \lambda - i \gamma) \varphi_n - \Omega \chi_n = \epsilon \left( \chi_{n+1} - 2 \chi_n + \chi_{n-1} \right) \\
\phantom{texttext} + 2 \left[ 2 |U_n|^2 (\psi_n + 2 \chi_n) + (U_n^2 + \bar{U}_n^2) (2 \varphi_n + \nu_n) \right], \\
(E - i \lambda + i \gamma) \psi_n - \Omega \nu_n = \epsilon \left( \nu_{n+1} - 2 \nu_n + \nu_{n-1} \right) \\
\phantom{texttext} + 2 \left[ 2 |U_n|^2 (\varphi_n + 2\nu_n) + (U_n^2 + \bar{U}_n^2) (2 \psi_n + \chi_n) \right], \\
(E + i \lambda + i \gamma) \chi_n - \Omega \varphi_n = \epsilon \left( \varphi_{n+1} - 2 \varphi_n + \varphi_{n-1} \right) \\
\phantom{texttext} + 2 \left[ 2 |U_n|^2 (2 \varphi_n + \nu_n) + (U_n^2 + \bar{U}_n^2) (\psi_n + 2 \chi_n) \right], \\
(E - i \lambda - i \gamma) \nu_n - \Omega \psi_n = \epsilon \left( \psi_{n+1} - 2 \psi_n + \psi_{n-1} \right) \\
\phantom{texttext} + 2 \left[ 2 |U_n|^2 (2 \psi_n + \chi_n) + (U_n^2 + \bar{U}_n^2) (\varphi_n + 2 \nu_n) \right],
\end{array} \right.
\end{equation}
where we have used the condition $V = \bar{U}$ for the $\PT$-symmetric breathers.
Recalling definition of the Hessian operator
$\mathcal{H}''_E$ in (\ref{PT-dNLS-var}), we can rewrite the spectral
problem (\ref{spectral-problem}) in the Hamiltonian form:
\begin{equation}
\mathcal{S} \mathcal{H}''_E \phi = i \lambda \phi,
\label{eq:eigenval}
\end{equation}
where $\mathcal{S}$ is a symmetric matrix with the blocks at each lattice node $n \in \mathbb{Z}$ given by
\begin{equation}
S := \left[ \begin{array}{cccc}
0 & 0 & 1 & 0 \\
0 & 0 & 0 & -1 \\
1 & 0 & 0 & 0 \\
0 & -1 & 0 & 0 \end{array} \right].
\end{equation}
We note the Hamiltonian symmetry of the eigenvalues of the spectral problem (\ref{eq:eigenval}).

\begin{proposition}
\label{prop-Ham-sym}
Eigenvalues of the spectral problem (\ref{eq:eigenval}) occur either as real or imaginary pairs
or as quadruplets in the complex plane.
\end{proposition}

\begin{proof}
Assume that $\lambda \in \mathbb{C}$ is an eigenvalue of the spectral problem (\ref{eq:eigenval}) with
the eigenvector $(\varphi,\psi,\chi,\nu)$. Then, $\bar{\lambda}$ is an eigenvalue of the same problem with
the eigenvector $(\bar{\psi},\bar{\varphi},\bar{\nu},\bar{\chi})$, whereas $-\lambda$ is also an eigenvalue
with the eigenvector $(\nu,\chi,\psi,\varphi)$.
\end{proof}

If $\Omega < -|\gamma|$ and $|E| < E_0 := \sqrt{\Omega^2 - \gamma^2}$ (points 2 and 4 shown on Figure \ref{branches}),
Theorem \ref{theorem-critical-point} implies that the self-adjoint operator $\mathcal{H}''_E$
in $\ell^2(\mathbb{Z})$ is negative-definite with the exception of either three (point 2) or one (point 4)
simple positive eigenvalues. In this case, we can apply the following Hamilton--Krein index
theorem in order to characterize the spectrum of $S \mathcal{H}''_E$. \\

{\bf Hamilton--Krein Index Theorem} (Theorem 3.3 in \cite{Kapitula}).
{\em Let $L$ be a self-adjoint operator in $\ell^2$ with finitely many negative
eigenvalues $n(L)$, a simple zero eigenvalue with eigenfunction $v_0$,
and the rest of its spectrum is bounded from below by a positive number.
Let $J$ be a bounded invertible skew-symmetric operator in $\ell^2$.
Let $k_r$ be a number of positive real eigenvalues of $JL$, $k_c$ be a number of quadruplets
$\{\pm\lambda,\pm\bar{\lambda}\}$ that are neither in $\mathbb{R}$ nor in $i\mathbb{R}$,
and $k_i^{-}$ be a number of purely imaginary pairs of eigenvalues of $JL$
whose invariant subspaces lie in the negative subspace of $L$.
Let $D = \langle L^{-1} J^{-1} v_0, J^{-1} v_0 \rangle_{\ell^2}$ be finite and nonzero. Then,
\begin{equation}
\label{index-Krein}
        K_{HAM} = k_r + 2k_c + 2k_i^{-} = \left\{ \begin{array}{lr} n(L) - 1, \quad & D < 0, \\
        n(L), \quad & D > 0. \end{array} \right.
\end{equation}
}

\vspace{0.25cm}

\begin{lemma}
\label{theorem-index}
Fix $\gamma \neq 0$, $\Omega < -|\gamma|$, and $0 < |E| < E_0$,
where $E_0 := \sqrt{\Omega^2 - \gamma^2} > 0$.
For every $\epsilon > 0$ sufficiently small, $K_{HAM} = 2$ for branch (b) of Lemma \ref{lemma-dimer}
and $K_{HAM} = 0$ for branch (c) of Lemma \ref{lemma-dimer} with $\Omega < -2 \sqrt{2} |\gamma|$.
For branch (c) with $\Omega \in (-2 \sqrt{2} |\gamma|,-|\gamma|)$,
there exists a value $E_s \in (0,E_0)$ such that $K_{HAM} = 1$
for $0 < |E| < E_s$ and $K_{HAM} = 0$ for $E_s < |E| < E_0$.
\end{lemma}

\begin{proof}
If $\gamma \neq 0$, $\Omega < -|\gamma|$, $|E| < E_0$, and $\epsilon > 0$ is sufficiently small,
Theorem \ref{theorem-critical-point} implies that the spectrum of $\mathcal{H}''_E$ in $\ell^2(\mathbb{Z})$
has finitely many positive eigenvalues and a simple zero eigenvalue with eigenvector $\sigma \Phi$.
Therefore, the Hamilton--Krein index theorem is applied in $\ell^2(\mathbb{Z})$ for $L = -\mathcal{H}''_E$,
$J = i \mathcal{S}$, and $v_0 = \sigma \Phi$. We shall verify that
\begin{equation}
\label{generalized-kernel-lemma}
\mathcal{H}''_E (\sigma \Phi) = 0, \quad \mathcal{S} \mathcal{H}''_E (\partial_E \Phi) = \sigma \Phi,
\end{equation}
where $\sigma \Phi$ is given by (\ref{kernel}) and $\partial_E \Phi$ denotes derivative of $\Phi$ with respect to parameter $E$.
The first equation $\mathcal{H}''_E (\sigma \Phi) = 0$ follows by Theorem \ref{theorem-critical-point}.
By differentiating equations (\ref{eq:statE}) in $E$, we obtain $\mathcal{H}''_E (\partial_E \Phi) = \mathcal{S} \sigma \Phi$
for every $E$, for which the solution $\Phi$ is differentiable in $E$. For $\epsilon = 0$, the limiting solution
of Lemma \ref{lemma-dimer} is differentiable in $E$ for every $E \neq 0$ and $E \neq \pm E_0$.
Due to smoothness of the continuation in $\epsilon$ by Theorem \ref{theorem-critical-point}, this property
holds for every $\epsilon > 0$ sufficiently small.

By using (\ref{generalized-kernel-lemma}) with $\mathcal{S}^{-1} = \mathcal{S}$, we obtain
\begin{equation}
D = -\langle (\mathcal{H}''_E)^{-1} \mathcal{S} \sigma \Phi, \mathcal{S} \sigma  \Phi \rangle_{\ell^2}
= -\langle \partial_E \Phi, \mathcal{S} \sigma \Phi \rangle_{\ell^2} =
-\sum_{n \in \mathbb{Z}} \partial_E \left( U_n \bar{V}_n + \bar{U}_n V_n \right)
= -\frac{d Q_{u,v}}{d E},
\end{equation}
where we have used the definition of $Q_{u,v}$ in (\ref{charge-PT-dNLS}).
We compute the slope condition at $\epsilon = 0$:
\begin{equation}
\label{slope-explicit}
\frac{d Q_{u,v}}{d E} \biggr|_{\epsilon = 0} = 2 \frac{d}{dE} \frac{A^2 E}{8A^2 + \Omega} =
4 (8A^2 + \Omega)  \frac{dA^2}{d E^2} \left[ 1 - \frac{\Omega \gamma^2}{(4 A^2 + \Omega)^3} \right],
\end{equation}
where relations (\ref{parameterization-A}) and (\ref{branch-E-A}) have been used.

For branch (b) of Lemma \ref{lemma-dimer} with $\Omega < -|\gamma|$,
we have $d A^2/d E^2 > 0$ and $|\Omega| < 4 A^2$, so that
$d Q_{u,v}/d E > 0$. By continuity, $d Q_{u,v}/d E $ remains strictly positive
for small $\epsilon > 0$. Thus, $D < 0$ and $K_{HAM} = 2$ by the Hamilton--Krein index theorem.

For branch (c) of Lemma \ref{lemma-dimer} with $\Omega < -|\gamma|$,
we have $d A^2/d E^2 < 0$ and $|\Omega| > 8 A^2$. Therefore, we only need to inspect
the sign of the expression $(4 A^2+\Omega)^3 - \Omega \gamma^2$.
If $\Omega < -2 \sqrt{2} |\gamma|$, then for every $A^2 \in (0,A_-^2)$, we have
$$
(4 A^2+\Omega)^3 - \Omega \gamma^2 \leq (4 A_-^2+\Omega)^3 - \Omega \gamma^2 = \frac{1}{8} \Omega^3 - \Omega \gamma^2
\leq \frac{1}{8} \Omega ( \Omega^2 - 8 \gamma^2) < 0,
$$
therefore, $D < 0$ and $K_{HAM} = 0$ by the Hamilton--Krein index theorem.

On the other hand, if $-2 \sqrt{2} |\gamma| < \Omega < -|\gamma|$,
we have $(4 A^2+\Omega)^3 - \Omega \gamma^2 < 0$ at $A = 0$ ($E = E_0$)
and $(4 A^2+\Omega)^3 - \Omega \gamma^2 > 0$ at $A = A_-$ ($E = 0$).
Since the dependence of $A$ versus $E$ is monotonic, there exists a value $E_s \in (0,E_0)$
such that $K_{HAM} = 1$ for $0 < |E| < E_s$ and $K_{HAM} = 0$ for $E_s < |E| < E_0$.
\end{proof}

If $K_{HAM} = 0$ and $D \neq 0$, orbital stability of a critical point
of $H_E$ in space $\ell^2(\mathbb{Z})$ can be proved from the
Hamilton--Krein theorem (see \cite{Kapitula} and references therein).
Orbital stability of breathers is understood in the following sense.

\begin{definition}
We say that the breather solution (\ref{stationary}) is orbitally stable in $\ell^2(\mathbb{Z})$
if for every $\nu > 0$ sufficiently small, there exists $\delta > 0$
such that if $\psi(0) \in \ell^2(\mathbb{Z})$ satisfies $\| \psi(0) - \Phi \|_{\ell^2} \leq \delta$,
then the unique global solution $\psi(t) \in \ell^2(\mathbb{Z})$, $t \in \mathbb{R}$ to the
$\PT$-symmetric dNLS equation (\ref{PT-dNLS}) satisfies the bound
\begin{equation}
\inf_{\alpha \in \mathbb{R}} \| e^{i \alpha} \psi(t) - \Phi \|_{\ell^2} \leq \nu, \quad \mbox{\rm for every} \;\; t \in \mathbb{R}.
\end{equation}
\label{def-stability}
\end{definition}

The definition of instability of breathers is given by negating Definition \ref{def-stability}.
The following result gives orbital stability or instability for branch (c) shown on Figure \ref{branches}.

\begin{theorem}
\label{theorem-bounded}
Fix $\gamma \neq 0$, $\Omega < -|\gamma|$, and $0 < |E| < E_0$.
For every $\epsilon > 0$ sufficiently small, the breather $(U,V)$
for branch (c) of Lemma \ref{lemma-dimer}
is orbitally stable in $\ell^2(\mathbb{Z})$ if $\Omega < -2 \sqrt{2} |\gamma|$.
For every $\Omega \in (-2 \sqrt{2} |\gamma|,-|\gamma|)$,
there exists a value $E_s \in (0,E_0)$ such that the breather $(U,V)$
is orbitally stable in $\ell^2(\mathbb{Z})$  if $E_s < |E| < E_0$ and unstable if $0 < |E| < E_s$.
\end{theorem}

\begin{proof}
The theorem is a corollary of Lemma \ref{theorem-index} for branch (c) of Lemma \ref{lemma-dimer}
and the orbital stability theory from \cite{Kapitula}.
\end{proof}

Orbital stability of breathers for branches (a) and (b) of Lemma \ref{lemma-dimer}
does not follow from the standard theory because $K_{HAM} = \infty$ for $|E| > E_0$
and $K_{HAM} = 2 > 0$ for branch (b) with $|E| < E_0$.
Nevertheless, by using smallness of parameter $\epsilon$ and the construction
of the breather $(U,V)$ in Theorem \ref{theorem-existence-soliton},
spectral stability of breathers can be considered directly.
Spectral stability and instability of breathers is understood in the following sense.

\begin{definition}
\label{def-spectral}
We say that the breather solution (\ref{stationary}) is spectrally stable if
$\lambda \in i\mathbb{R}$ for every bounded solution of the spectral problem (\ref{eq:eigenval}).
On the other hand, if the spectral problem (\ref{eq:eigenval}) admits an eigenvalue
$\lambda \notin i \mathbb{R}$ with an eigenvector in $\ell^2(\mathbb{R})$, we say
that the breather solution (\ref{stationary}) is spectrally unstable.
\end{definition}

The following theorem gives spectral stability of breathers
for branches (a) and (b) shown on Figure \ref{branches}.

\begin{theorem}
\label{theorem-stability}
Fix $\gamma \neq 0$, $|\Omega| > |\gamma|$, and $E$ along branches (a) and (b) of Lemma
\ref{lemma-dimer} with $E \neq 0$ and $E \neq \pm E_0$. For every $\epsilon > 0$ sufficiently small,
the spectral problem (\ref{eq:eigenval})
admits a double zero eigenvalue with the generalized eigenvectors
\begin{equation}
\label{generalized-kernel}
\mathcal{H}''_E (\sigma \Phi) = 0, \quad \mathcal{S} \mathcal{H}''_E (\partial_E \Phi) = \sigma \Phi,
\end{equation}
where the eigenvector $\sigma \Phi$ is given by (\ref{kernel}) and
the generalized eigenvector $\partial_E \Phi$ denotes derivative of $\Phi$ with respect to parameter $E$.
For every $E$ such that the following non-degeneracy condition is satisfied,
\begin{equation}
\label{assumption-non-degeneracy}
2 \sqrt{(4A^2 + \Omega)^2 - \frac{\Omega \gamma^2}{4 A^2 + \Omega}} \neq E \pm \sqrt{\Omega^2 - \gamma^2},
\end{equation}
the breather $(U,V)$ is spectrally stable.
\end{theorem}

\begin{proof}
If $\epsilon = 0$, the breather solution of Theorem \ref{theorem-existence-soliton}
is given by $U_n = 0$ for every $n \neq 0$ and $U_0 = A e^{i \theta}$,
where $A$ and $\theta$ are defined by Lemma \ref{lemma-dimer}.
In this case, the spectral problem (\ref{spectral-problem}) decouples
into $4$-by-$4$ blocks for each lattice node $n \in \mathbb{Z}$.
Recall that $\mathcal{H}_E'' = \mathcal{L}$ at $\epsilon = 0$.

For $n = 0$, eigenvalues $\lambda$ are determined by the $4$-by-$4$ matrix block $-i S \mathcal{L}_0$.
Using relations (\ref{parameterization-A}) and (\ref{branch-E-A}), as well as
symbolic computations with MAPLE, we found that the $4$-by-$4$ matrix block $-i S \mathcal{L}_0$ has
a double zero eigenvalue and a pair of simple eigenvalues at $\lambda = \pm \lambda_0$, where
\begin{equation}
\label{lambda-0-eig}
\lambda_0 = 2 i \sqrt{(4A^2 + \Omega)^2 - \frac{\Omega \gamma^2}{4 A^2 + \Omega}}.
\end{equation}

For $n \in \mathbb{Z} \backslash \{0\}$, eigenvalues $\lambda$ are determined by the $4$-by-$4$ 
matrix block $-i S \mathcal{L}_n$, where $\mathcal{L}_n$ is given by (\ref{block-L-n}).
If $|\gamma| < |\Omega|$, $E \neq 0$, and $E \neq \pm E_0$, where $E_0 := \sqrt{\Omega^2 - \gamma^2}$,
each block has four simple eigenvalues $\pm \lambda_+$ and $\pm \lambda_-$, where
\begin{equation}
\label{lambda-plus-minus-eig}
\lambda_{\pm} := i ( E \pm E_0),
\end{equation}
so that $\lambda_{\pm} \in i \mathbb{R}$. Since there are infinitely many nodes with $n \neq 0$,
the four eigenvalues are semi-simple and have infinite multiplicity.

If $\epsilon > 0$ is sufficiently small, we use perturbation theory for linear operators from Section \ref{variational}.

\begin{itemize}
\item The double zero eigenvalue persists with respect to $\epsilon \neq 0$ at zero
because of the gauge invariance of the breather $(U,V)$ (with respect to rotation of the complex phase).
Indeed, $\mathcal{H}''_E (\sigma \Phi) = 0$ follows from the result of Theorem \ref{theorem-critical-point}.
The generalized eigenvector is defined by equation $\mathcal{S} \mathcal{H}''_E \Psi = \sigma \Phi$,
which is equivalent to equation $\mathcal{H}''_E \Psi = (V,\bar{V},U,\bar{U})^T$.
Differentiating equations (\ref{eq:statE}) in $E$, we obtain $\Psi = \partial_E \Phi$.
Since ${\rm dim}[{\rm Ker}(\mathcal{H}_E'')] = 1$ and
\begin{equation}
\langle \sigma \Phi, \mathcal{S} \partial_E \Phi \rangle_{\ell^2} =
\sum_{n \in \mathbb{Z}} \partial_E \left( U_n \bar{V}_n + \bar{U}_n V_n \right)
= \frac{d Q_{u,v}}{d E},
\end{equation}
the second generalized eigenvector $\tilde{\Psi} \in \ell^2(\mathbb{Z})$
exists as a solution of equation $\mathcal{S} \mathcal{H}''_E \tilde{\Psi} = \partial_E \Phi$ if and only if $d Q_{u,v}/d E = 0$.
It follows from the explicit computation (\ref{slope-explicit}) that if $\epsilon = 0$, then $d Q_{u,v}/d E \neq 0$
for every $E$ along branches (a) and (b) of Lemma \ref{lemma-dimer}.
By continuity, $d Q_{u,v}/d E \neq 0$ for small $\epsilon > 0$.
Therefore, the zero eigenvalue of the operator $-i \mathcal{S} \mathcal{H}''_E$ is exactly double for small $\epsilon > 0$.

\item Using the same computation (\ref{slope-explicit}), it is clear that $\lambda_0 \in i \mathbb{R}$
for every $E$ along branches (a) and (b) of Lemma \ref{lemma-dimer}.
Assume that $\lambda_0 \neq \pm \lambda_+$ and $\lambda_0 \neq \pm \lambda_-$,
which is expressed by the non-degeneracy condition (\ref{assumption-non-degeneracy}). Then,
the pair $\pm \lambda_0$ is isolated from the rest of the spectrum of the operator $-i \mathcal{S} \mathcal{H}''_E$ at $\epsilon = 0$.
Since the eigenvalues $\lambda = \pm \lambda_0$ are simple and purely imaginary, they
persist on the imaginary axis for $\epsilon \neq 0$ because they cannot
leave the imaginary axis by the Hamiltonian symmetry of Proposition \ref{prop-Ham-sym}.

\item If $|\gamma| < |\Omega|$, $E \neq 0$, and $E \neq \pm E_0$,
the semi-simple eigenvalues $\pm \lambda_+$ and $\pm \lambda_-$ of infinite multiplicity
are nonzero and located at the imaginary axis at different points for $\epsilon = 0$.
They persist on the imaginary axis for $\epsilon \neq 0$ according to the following perturbation argument.
First, for the central site $n = 0$, the spectral problem (\ref{spectral-problem}) can be
written in the following abstract form
$$
\left( S \mathcal{L}_0(\epsilon) - 2 \epsilon S - i \lambda I \right) \phi_0 = -\epsilon S(\phi_1 + \phi_{-1}),
$$
where $\mathcal{L}_0(\epsilon)$ denotes a continuation of $\mathcal{L}_0$ in $\epsilon$.
Thanks to the non-degeneracy condition (\ref{assumption-non-degeneracy}) as well as
the condition $\lambda_{\pm} \neq 0$, the matrix $S \mathcal{L}_0 - i \lambda_{\pm} I$ is invertible. 
By continuity, the matrix $S \mathcal{L}_0(\epsilon) - i \lambda I$ is invertible for 
every $\epsilon$ and $\lambda$ near $\epsilon = 0$ and $\lambda = \lambda_{\pm}$. 
Therefore, there is a unique $\phi_0$ given by
$$
\phi_0 = -\epsilon \left( S \mathcal{L}_0(\epsilon) - 2 \epsilon S - i \lambda I \right)^{-1} S(\phi_1 + \phi_{-1}),
$$
which satisfies $| \phi_0 | \leq C \epsilon (|\phi_1| + |\phi_{-1}|)$ near $\epsilon = 0$ and $\lambda = \lambda_{\pm}$,
where $C$ is a positive $\epsilon$- and $\lambda$-independent constant.
Next, for either $n \in \mathbb{N}$ or $-n \in \mathbb{N}$, the spectral problem (\ref{spectral-problem})
can be represented in the form
$$
S \mathcal{L}_n(\epsilon) \phi_n + \epsilon S (\Delta \phi)_n - i \lambda \phi_n = - \delta_{n,\pm 1} \epsilon S \phi_0, \quad \pm n \in \mathbb{N},
$$
where $\mathcal{L}_n(\epsilon)$ denotes a continuation of $\mathcal{L}_n$ given by (\ref{block-L-n}) in $\epsilon$,
whereas the operator $\Delta$ is applied with zero end-point condition at $n = 0$.
We have $\mathcal{L}_n(\epsilon) = \mathcal{L}_n + \mathcal{O}(\epsilon^2)$ and $\epsilon S \phi_0 = \mathcal{O}(\epsilon^2)$
near $\epsilon = 0$ and $\lambda = \lambda_{\pm}$. Therefore, up to the first order of the perturbation theory,
the spectral parameter $\lambda$ near $\lambda_{\pm}$ is defined from the truncated eigenvalue
problem
\begin{equation}
\label{truncated-problem}
S \mathcal{L}_n \phi_n + \epsilon S (\Delta \phi)_n = i \lambda \phi_n, \quad \pm n \in \mathbb{N},
\end{equation}
which is solved with the discrete Fourier transform (\ref{discrete-Fourier}).
In order to satisfy the Dirichlet end-point condition at $n = 0$, the sine--Fourier transform
must be used, which does not affect the characteristic equation for the purely continuous spectrum
of the spectral problem (\ref{truncated-problem}). By means of routine computations, we obtain the
characteristic equation in the form, see also equation (\ref{dispersion-relation}),
\begin{equation}
\label{continuous-spectrum}
(E \pm i \lambda)^2 + \gamma^2 - \left( \Omega - 4 \epsilon \sin^2 \frac{k}{2} \right)^2 = 0,
\end{equation}
where $k \in [-\pi,\pi]$ is the parameter of the discrete Fourier transform (\ref{discrete-Fourier}).
Solving the characteristic equation (\ref{continuous-spectrum}), we obtain four branches of the continuous spectrum
\begin{equation}
\lambda = \pm i \left( E \pm \sqrt{\left(\Omega - 4 \epsilon \sin^2 \frac{k}{2}\right)^2 - \gamma^2} \right),
\end{equation}
where the two sign choices are independent from each other.
If $|\Omega| > |\gamma|$ is fixed and $\epsilon > 0$ is small, the four branches of the continuous spectrum
are  located on the imaginary axis near the points $\pm \lambda_+$ and $\pm \lambda_-$ given by (\ref{lambda-plus-minus-eig}).

In addition to the continuous spectrum given by (\ref{continuous-spectrum}), there may exist isolated
eigenvalues near $\pm \lambda_+$ and $\pm \lambda_-$, which are found from the second-order 
perturbation theory \cite{PelSak}. Under the condition $E \neq 0$ and $E \neq \pm E_0$, 
these eigenvalues are purely imaginary. Therefore,
the infinite-dimensional part of the spectrum of the operator $-i \mathcal{S} \mathcal{H}^{''}_E$
persists on the imaginary axis for $\epsilon \neq 0$ near the points $\pm \lambda_+$ and $\pm \lambda_-$
of infinite algebraic multiplicity.
\end{itemize}

The statement of the lemma follows from the perturbation theory and
the fact that all isolated eigenvalues and the continuous spectrum of
$-i \mathcal{S} \mathcal{H}_E''$ are purely imaginary.
\end{proof}

\begin{remark}
In the asymptotic limit $E^2 = 64 A^4 + \mathcal{O}(A^2)$ as $A \to \infty$,
the eigenvalues $\lambda_0$ and $\lambda_{\pm}$ defined by (\ref{lambda-0-eig}) and (\ref{lambda-plus-minus-eig}) 
are given asymptotically by 
\begin{eqnarray}
\lambda_0 \approx i |E|, \quad \lambda_+ \approx i E, \quad \lambda_- \approx i E,
\end{eqnarray}
where the remainder terms are $\mathcal{O}(1)$ as $|E| \to \infty$. The values
$\lambda_0$, $\lambda_+$, and $\lambda_-$ are close to each other as $E \to +\infty$.
\end{remark}

\begin{remark}
Computations in the proof of Theorem \ref{theorem-stability} can be extended to the branch
(c) of Lemma \ref{lemma-dimer}. Indeed, $\lambda_0 \in i \mathbb{R}$
for branch (c) with either $\Omega < -2 \sqrt{2} |\gamma|$ or $\Omega \in (-2 \sqrt{2} |\gamma|,-|\gamma|)$,
and $E$ near $\pm E_0$. On the other hand, $\lambda_0 \in \mathbb{R}$
if $\Omega \in (-2 \sqrt{2} |\gamma|,-|\gamma|)$ and $E$ near $0$.
As a result, branch (c) is spectrally stable in the former case and is spectrally unstable
in the latter case, in agreement with Theorem \ref{theorem-bounded}.
\end{remark}

\begin{remark}
Observe in the proof of Theorem \ref{theorem-stability} that $\lambda_{\pm} \notin i \mathbb{R}$
if $|\Omega| < |\gamma|$. In this case, branch (b) of Lemma \ref{lemma-dimer}
is spectrally unstable. This instability corresponds to the instability of the zero equilibrium for
$|\Omega| < |\gamma|$, in agreement with the result of Proposition \ref{prop-zero-equilibrium}.
\end{remark}

Before presenting numerical approximations of eigenvalues of the spectral
problem (\ref{eq:eigenval}), we compute the Krein signature of wave continuum.
This helps to interpret instabilities and resonances that arise when
isolated eigenvalues $\pm \lambda_0$ cross the continuous bands near points
$\pm \lambda_+$ and $\pm \lambda_-$. The Krein signature of simple isolated eigenvalues is
defined as follows.

\begin{definition}
\label{def-Krein}
Let $\phi \in \ell^2(\mathbb{Z})$ be an eigenvector of the spectral problem (\ref{eq:eigenval})
for an isolated simple eigenvalue $\lambda_0 \in i \mathbb{R}$. Then, the energy quadratic form
$\langle \mathcal{H}_E'' \phi, \phi \rangle_{\ell^2}$ is nonzero and its sign
is called the Krein signature of the eigenvalue $\lambda_0$.
\end{definition}

Definition \ref{def-Krein} is used to simplify the presentation. Similarly, one can define the Krein signature
of isolated multiple eigenvalues and the Krein signature of the continuous spectral bands
in the spectral problem (\ref{eq:eigenval}) \cite{Kapitula}. The following lemma
characterizes Krein signatures of the spectral points arising in the proof of Theorem \ref{theorem-stability}.

\begin{lemma}
\label{lemma-Krein}
Fix $\gamma \neq 0$, $|\Omega| > |\gamma|$, and $E > 0$ with $E \neq \pm E_0$.
Assume the non-degeneracy condition (\ref{assumption-non-degeneracy}).
For every $\epsilon > 0$ sufficiently small, we have the following for
the corresponding branches of Lemma \ref{lemma-dimer}:
\begin{itemize}
\item[(a)] the subspaces of $-i \mathcal{S} \mathcal{H}''_E$ in $\ell^2(\mathbb{Z})$
near $\pm \lambda_+$, $\pm \lambda_-$, and $\pm \lambda_0$ have
positive, negative, and positive Krein signature, respectively;
\item[(b)] the subspaces of $-i \mathcal{S} \mathcal{H}''_E$ in $\ell^2(\mathbb{Z})$
near $\pm \lambda_+$, $\pm \lambda_-$, and $\pm \lambda_0$ have
negative, positive (if $E > E_0$) or negative (if $E < E_0$), and positive Krein signature, respectively;
\item[(c)] all subspaces of $-i \mathcal{S} \mathcal{H}''_E$ in $\ell^2(\mathbb{Z})$
near $\pm \lambda_+$, $\pm \lambda_-$, and $\pm \lambda_0$ (if $\lambda_0 \in i \mathbb{R}$) 
have negative Krein signature.
\end{itemize}
\end{lemma}

\begin{proof}
We proceed by the perturbation arguments from the limit $\epsilon = 0$, where
$-i \mathcal{S} \mathcal{H}''_E = - i \mathcal{S} \mathcal{L}$ is a block-diagonal operator consisting of
$4 \times 4$ blocks. In particular, we consider the blocks for $n \in \mathbb{Z} \backslash \{0\}$,
where $\mathcal{L}_n$ is given by (\ref{block-L-n}).
Solving (\ref{spectral-problem}) at $\epsilon = 0$ and $\lambda = \lambda_{\pm}$, we
obtain the eigenvector
$$
\varphi_n = -\Omega, \quad \psi_n = 0, \quad \chi_n = \pm E_0 + i \gamma, \quad \nu_n = 0, \quad n \in \mathbb{Z} \backslash \{0\}.
$$
As a result, we obtain for the eigenvector $\phi_n = (\varphi_n, \psi_n, \chi_n,\nu_n)$,
\begin{eqnarray*}
K_n := \langle \mathcal{L}_n \phi_n, \phi_n \rangle_{\ell^2} & = & \Omega (|\varphi_n|^2 + |\chi_n|^2) - (E + i \gamma) \chi_n \bar{\varphi}_n - (E-i \gamma) \varphi_n \bar{\chi}_n \\
& = & 2 \Omega E_0 (E_0 \pm E).
\end{eqnarray*}

For branch (a), $\Omega > |\gamma|$ and $E > E_0$. Therefore, $K_n > 0$
for $\lambda = \lambda_+$ and $K_n < 0$ for $\lambda = \lambda_-$.

For branch (b), $\Omega < -|\gamma|$ and either $E > E_0$ or $E \in (0,E_0)$. In either case,
$K_n < 0$ for $\lambda = \lambda_+$. On the other hand, for $\lambda = \lambda_-$, $K_n > 0$
if $E > E_0$ and $K_n < 0$ if $E \in (0,E_0)$.

For branch (c), $\Omega < -|\gamma|$ and $E \in (0,E_0)$. In this case, $K_n < 0$ for either
$\lambda = \lambda_+$ or $\lambda = \lambda_-$.

Finally, the Krein signature for the eigenvalue $\lambda = \lambda_0$ denoted by $K_0$ follows
from the computations of eigenvalues $\mu_{1,2,3}$ in the proof of Theorem \ref{theorem-critical-point}.
We have $K_0 > 0$ for branches (a) and (b) because $\mu_{1,2,3} > 0$ and
we have $K_0 < 0$ for branch (c) because $\mu_{1,3} < 0$, whereas the eigenvalue $\mu_2 > 0$
is controlled by the result of Lemma \ref{theorem-index}.

The signs of all eigenvalues are nonzero and continuous with respect to parameter $\epsilon$.
Therefore, the count above extends to the case of small nonzero $\epsilon$.
\end{proof}

The spectrum of $-i \mathcal{S} \mathcal{H}_E''$ is shown at Figure~\ref{spectra}.
Panels (a), (b) and (c) correspond to branches shown at Figure~\ref{branches}.

\begin{figure}
\centering
\includegraphics[scale=0.45]{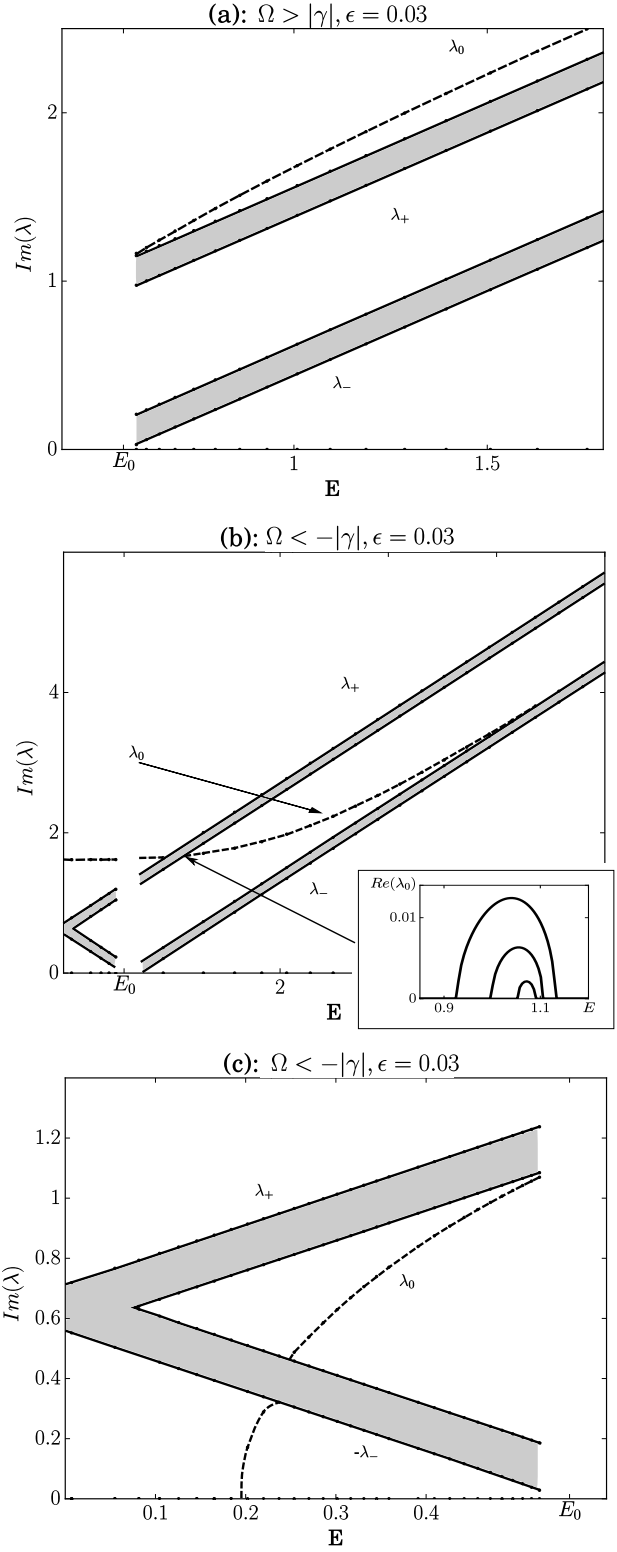}
\caption{The spectrum of $-i \mathcal{S} \mathcal{H}_E''$
for different branches of breathers.}
\label{spectra}
\end{figure}

\begin{itemize}
\item[(a)] We can see on panel (a) of Figure \ref{spectra} that
$\lambda_0,\lambda_{\pm}$ do not intersect for every $E > E_0$ and are
located within fixed distance $\mathcal{O}(1)$, as $|E|\to\infty$.
Note that the upper-most $\lambda_0$ and $\lambda_+$ have positive Krein signature,
whereas the lowest $\lambda_-$ has negative Krein signature, as is given by Lemma \ref{lemma-Krein}.

\item[(b)] We observe on panel (b) of Figure \ref{spectra} that
$\lambda_{+}$ intersects $\lambda_0$, creating a small bubble of instability in
the spectrum. The insert shows that the bubble shrinks as $\epsilon \to 0$,
in agreement with Theorem \ref{theorem-stability}. There is also an
intersection between $\lambda_{-}$ and $\lambda_0$, which does not create instability.
These results are explained by the Krein signature computations in Lemma \ref{lemma-Krein}.
Instability is induced by opposite Krein signatures between $\lambda_+$ and $\lambda_0$,
whereas crossing of $\lambda_-$ and $\lambda_0$ with the same Krein signatures is safe
of instabilities. Note that for small $E$, the isolated eigenvalue $\lambda_0$ is
located above both the spectral bands near $\lambda_+$ and $\lambda_-$.
The gap in the numerical data near $E = E_0$ indicates failure to continue the
breather solution numerically in $\epsilon$, in agreement with the proof
of Theorem \ref{theorem-existence-soliton}.

\item[(c)] We observe from panel (c) of Figure \ref{spectra} that
$\lambda_0$ and $-\lambda_{-}$ intersect but do not create instabilities, since
all parts of the spectrum have the same signature, as is given by Lemma \ref{lemma-Krein}.
In fact, the branch is both spectrally and orbitally stable as long as $\lambda_0 \in i \mathbb{R}$,
in agreement with Theorem \ref{theorem-bounded}.
On the other hand, there is $E_s \in (0,E_0)$, if $\Omega \in (-2 \sqrt{2} |\gamma|,-|\gamma|)$,
such that $\lambda_0 \in \mathbb{R}$ for $E \in (0,E_s)$, which indicates instability of branch (c),
again, in agreement with Theorem \ref{theorem-bounded}.
\end{itemize}

As we see on panel (b) of Figure \ref{spectra}, $\lambda_0$ intersects $\lambda_+$ for some $E=E_*>E_0$.
In the remainder of this section, we study whether this crossing point is always located on the right
of $E_0$. In fact, the answer to this question is negative. We shall prove for branch (b)
that the intersection of $\lambda_0$ with either $\lambda_+$ or $-\lambda_-$ occur either for $E_*>E_0$ or
for $E_*<E_0$, depending on parameters $\gamma$ and $\Omega$.

\begin{lemma}
Fix $\gamma \neq 0$, $\Omega < -|\gamma|$, and $E > 0$ along branch (b) of Lemma~\ref{lemma-dimer}.
There exists a resonance $\lambda_0 = \lambda_+$ at $E = E_*$ with $E_* > E_0$ if $\Omega \in (\Omega_*,-|\gamma|)$
and $E_* \in (0,E_0)$ if $\Omega \in (-5|\gamma|,\Omega_*)$, where
\begin{equation}
\label{omega-star}
\Omega_* := -\frac{\sqrt{1+5\sqrt{2}}}{\sqrt{2}}|\gamma|.
\end{equation}
Moreover, if $\Omega < -5 |\gamma|$, there exists a resonance $\lambda_0 = -\lambda_-$ at $E = E_*$
with $E_* \in (0,E_0)$.
\end{lemma}

\begin{proof}
Let us first assume that there exists a resonance $\lambda_0 = \lambda_+$ at $E = E_* = E_0$
and find the condition on $\gamma$ and $\Omega$, when this is possible.
From the definitions (\ref{lambda-0-eig}) and (\ref{lambda-plus-minus-eig}), we obtain
the constraint on $A^2$:
\begin{equation*}
(4A^2 + \Omega)^2 - \frac{\Omega\gamma^2}{4A^2+\Omega} = E_0^2 = \Omega^2 - \gamma^2.
\end{equation*}
After canceling $4 A^2$ since $A^2 \geq A_+^2 > 0$ with $A_+^2$ given by (\ref{A-plus}),
we obtain
$$
16 A^4 + 12 \Omega A^2 + 2 \Omega^2 + \gamma^2 = 0,
$$
which has two roots
\begin{equation*}
A^2 = -\frac{3}{8}\Omega \pm \frac18 \sqrt{\Omega^2 - 4\gamma^2}.
\end{equation*}
Since $A^2 \geq A_+^2$, the lower sign is impossible because this leads to a contradiction
$$
\sqrt{|\Omega| - 2 |\gamma|} - \sqrt{|\Omega| + 2 |\gamma|} \geq 0.
$$
The upper sign is possible if $|\Omega| \geq 2|\gamma|$.
Using the parametrization~\eqref{branch-E-A}, we substitute the root for $A^2$ to the
equation $E_0^2 = E^2$ and simplify it:
\begin{eqnarray*}
\Omega^2 - \gamma^2 & = & \left(2 |\Omega| + \sqrt{\Omega^2 - 4 \gamma^2}\right)^2 \left[ 1 - \frac{4 \gamma^2}{\left(|\Omega| +
\sqrt{\Omega^2 - 4 \gamma^2}\right)^2} \right] \\
&=& \frac{2 \sqrt{\Omega^2 - 4 \gamma^2} \left(2 |\Omega| + \sqrt{\Omega^2 - 4 \gamma^2}\right)^2}{\left(|\Omega| +
\sqrt{\Omega^2 - 4 \gamma^2}\right)}.
\end{eqnarray*}
This equation further simplifies to the form:
\begin{equation*}
\sqrt{\Omega^2 - 4\gamma^2} (9\Omega^2 - 7\gamma^2) + \Omega (31\gamma^2 - 7\Omega^2) = 0.
\end{equation*}
Squaring it up, we obtain
\begin{equation*}
        8\Omega^6 - 4\Omega^4 \gamma^2 - 102 \Omega^2 \gamma^4 - 49 \gamma^6= 0,
\end{equation*}
which has only one positive root for $\Omega^2$ given by
\begin{equation*}
\Omega^2 = \frac{1 +  5\sqrt{2}}{2} \gamma^2.
\end{equation*}
This root yields a formula for $\Omega_*$ in (\ref{omega-star}).
Since there is a unique value for $\Omega \in (-\infty,-|\gamma|)$,
for which the case $E_* = E_0$ is possible,
we shall now consider whether $E_* > E_0$ or $E_* < E_0$ for
$\Omega \in (-\infty,\Omega_*)$ or $\Omega \in (\Omega_*,-|\gamma|)$.

To inspect the range $E_* < E_0$, we consider a particular case,
for which the intersection $\lambda_0 = \lambda_+ = -\lambda_-$
happens at $E = 0$. In this case, $A^2 = A_+^2$ given by (\ref{A-plus}),
so that the condition $\lambda_0^2 = -E_0^2$ can be rewritten as
\begin{equation*}
4(\gamma^2 - |\gamma| \Omega) = \Omega^2-\gamma^2.
\end{equation*}
There is only one negative root for $\Omega$ and it is given by $\Omega = -5 |\gamma|$.
By continuity, we conclude that $\lambda_0 = \lambda_+$ for $\Omega \in (-5 |\gamma|, \Omega_*)$
and $\lambda_0 = -\lambda_-$ for $\Omega \in (-\infty,-5|\gamma|)$, both
cases correspond to $E_* \in (0,E_0)$.

Finally, we verify that the case $\lambda_0 = \lambda_+$ occurs for $E_* > E_0$ if $\Omega \in (\Omega_*,-|\gamma|)$.
Indeed, $\lambda_0 = i (8 A^2 + 2 \Omega + \mathcal{O}(A^{-2}))$ and $\lambda_+ = i( 8 A^2 + \Omega + E_0 + \mathcal{O}(A^{-2}))$ 
as $A \to \infty$, so that ${\rm Im}(\lambda_0) < {\rm Im}(\lambda_+)$ as $E \to \infty$.
On the other hand, the previous estimates suggest that 
${\rm Im}(\lambda_0) > {\rm Im}(\lambda_+)$ for every $E \in (0,E_0)$
if $\Omega \in (\Omega_*,-|\gamma|)$. Therefore, there exists at least one intersection
$\lambda_0 = \lambda_+$ for $E_* > E_0$ if $\Omega \in (\Omega_*,-|\gamma|)$.
\end{proof}

\section{Summary}
\label{conclusion}

We have reduced Newton's equation of motion for coupled pendula shown on Figure \ref{pendulapic}
under a resonant periodic force to the $\PT$-symmetric dNLS equation (\ref{PT-dNLS}).
We have shown that this system is Hamiltonian with conserved energy (\ref{energy-PT-dNLS}) and an additional
constant of motion (\ref{charge-PT-dNLS}). We have studied breather solutions of this model,
which generalize symmetric synchronized oscillations of coupled pendula that arise if $E = 0$.
We showed existence of three branches of breathers shown on Figure \ref{branches}. We also
investigated their spectral stability analytically and numerically. The spectral information
on each branch of solutions is shown on Figure \ref{spectra}. For branch (c), we were also
able to prove orbital stability and instability from the energy method.
The technical results of this paper are summarized in Table 1 and described as follows.

\begin{table}[htb]
\begin{minipage}{\textwidth}
\centering
\caption{A summary of results on breather solutions for small $\epsilon$. Here, IB is a narrow
instability bubble seen on panel (b) of Figure \ref{spectra}.}
\vspace{1em}
\begin{tabular}{|c|c|c|c|c|}
        \hline
        & \multicolumn{2}{|c|}{$|E|>E_0$} & \multicolumn{2}{|c|}{$|E|<E_0$} \\
        \cline{2-5}
        \scell{Parameter \\intervals}        & $\Omega > |\gamma|$
                        & $\Omega < -|\gamma|$
                        & $\Omega < -|\gamma|$
                        & $\Omega < -|\gamma|$ \\ \hline
        Existence & point 1 & point 2 & point 3 & point 4 \\
        on Figure \ref{branches} & on branch (a) & on branch (b) & on branch (b) & on branch (c) \\ \hline
        Continuum & Sign-indefinite & Sign-indefinite & Negative & Negative \\ \hline
        \scell{Spectral \\ stability} & Yes & Yes (IB) & Yes (IB)
        & \scell{Depends\\ on parameters} \\ \hline
\scell{Orbital \\ stability} & No & No & \scell{Yes \\ if $|\lambda_0| > |\lambda_{\pm}|$} & \scell{Yes \\ if spectrally stable} \\ \hline
\end{tabular}
\end{minipage}
\label{table-summary}
\end{table}

For branch (a), we found that it is disconnected from the symmetric synchronized oscillations at $E = 0$.
Along this branch, breathers of small amplitudes $A$ are connected to breathers of large amplitudes $A$. Every
point on the branch corresponds to the saddle point of the energy function between two wave continua of positive
and negative energies. Every breather along the branch is spectrally stable and is free of resonance between
isolated eigenvalues and continuous spectrum. In the follow-up work \cite{ChernPel}, we will prove
long-time orbital stability of breathers along this branch.

For branch (b), we found that the large-amplitude breathers as $E \to \infty$ are connected to the symmetric
synchronized oscillations at $E = 0$, which have the smallest (but nonzero) amplitude $A = A_+$.
Breathers along the branch are spectrally stable except for a narrow instability bubble, where
the isolated eigenvalue $\lambda_0$ is in resonance with the continuous spectrum.
The instability bubble can occur either for $E > E_0$, where the breather is a saddle point of the
energy function between two wave continua of opposite energies or for $E < E_0$, where the breather
is a saddle point between the two negative-definite wave continua and directions of positive energy.
When the isolated eigenvalue of positive energy $\lambda_0$ is above the continuous spectrum near
$\lambda_+$ and $\pm \lambda_-$, orbital stability of breathers can be proved by using the technique 
in \cite{Cuccagna}, which was developed for the dNLS equation.

Finally, for branch (c), we found that the small-amplitude breathers at $E \to E_0$ are connected
to the symmetric synchronized oscillations at $E = 0$, which have the largest amplitude $A = A_-$.
Breathers are either spectrally stable near $E = E_0$ or unstable near $E = 0$, depending
on the detuning frequency $\Omega$ and the amplitude of the periodic resonant force $\gamma$.
When breathers are spectrally stable, they are also orbitally stable for infinitely long times.

\vspace{0.25cm}

\noindent{\bf Acknowledgements.}
The authors thank I.V. Barashenkov and P.G. Kevrekidis for useful discussions regarding this project.
The work of A.C. is supported by the graduate scholarship at McMaster University.
The work of D.P. is supported by the Ministry of Education
and Science of Russian Federation (the base part of the State task No. 2014/133, project No. 2839).

\end{document}